\documentclass[AMA,STIX1COL]{WileyNJD-v2}

\articletype{Article Type}%

\received{26 April 2016}
\revised{6 June 2016}
\accepted{6 June 2016}

\makeatletter  
\newif\if@restonecol  
\makeatother

\usepackage[linesnumbered,ruled,vlined, algo2e]{algorithm2e}%[ruled,vlined]{  
\usepackage{caption}
\usepackage{subcaption}
\usepackage{graphicx}
\usepackage{amsmath}
\usepackage{amssymb}
\usepackage{textcomp}
\usepackage{booktabs}
\usepackage[justification=justified]{caption}

\usepackage{color} 
\usepackage{float}
\usepackage{hyperref} 
\usepackage{url}
\usepackage{multirow}
\begin{document}
	\title{PDMA: Probabilistic Service Migration Approach for Delay-aware and Mobility-aware Mobile Edge Computing}
	
	\author[1]{Minxian Xu}

\author[1]{Qiheng Zhou*}

\author[2]{Huaming Wu}

\author[3]{Weiwei Lin}

\author[1]{Kejiang Ye}

\author[4]{Chengzhong Xu}
	
\authormark{M. Xu \textsc{et al}}	

\address[1]{\orgdiv{Shenzhen Institutes of Advanced Technology, Chinese Academy of Sciences}, \orgaddress{\state{Shenzhen}, \country{China}}}

\address[2]{\orgdiv{Center for Applied Mathematics, Tianjin University}, \orgaddress{\state{Tianjin}, \country{China}}}

\address[3]{\orgdiv{School of Computer Science and Engineering, South China University of Technology}, \orgaddress{\state{Guangzhou}, \country{China}}}

\address[4]{\orgdiv{State Key Lab of IOTSC, Department of Computer Science}, \orgname{University of Macau}, \orgaddress{\state{Macau SAR}, \country{China}}}

%\address[3]{\orgdiv{Org Division}, \orgname{Org Name}, \orgaddress{\state{State name}, \country{Country name}}}

% \corres{*Qiheng Zhou (\email{qh.zhou@siat.ac.cn) }}
\corres{*Qiheng Zhou (\email{zhouqh7@outlook.com) }}

\abstract[Summary]{As a key technology in the 5G era, Mobile Edge Computing (MEC) has developed rapidly in recent years. MEC aims to reduce the service delay of mobile users, while alleviating the processing pressure on the core network. MEC can be regarded as an extension of cloud computing on the user side, which can deploy edge servers and bring computing resources closer to mobile users, and provide more efficient interactions. However, due to the user's dynamic mobility, the distance between the user and the edge server will change dynamically, which may cause fluctuations in Quality of Service (QoS). Therefore, when a mobile user moves in the MEC environment, certain approaches are needed to schedule services deployed on the edge server to ensure the user experience. In this paper, we model service scheduling in MEC scenarios and propose a delay-aware and mobility-aware service management approach based on concise probabilistic methods. This approach has low computational complexity and can effectively reduce service delay and migration costs. Furthermore, we conduct experiments by utilizing multiple realistic datasets and use iFogSim to evaluate the performance of the algorithm. 
%The experiment results on rush hour traces show that the algorithm proposed in this work optimizes the performance on service delay and migration cost with \color{red} 8\% to 20\% and over 75\% improvement, respectively, compared with baselines. 
The results show that our proposed approach can optimize the performance on service delay, with 8\% to 20\% improvement and reduce the migration cost by more than 75\% compared with baselines during the rush hours.
%over 75\% migration cost compared with baselines during the rush hours.
}


\maketitle

\section{Introduction}
Driven by the fast growth of the Internet of Things (IoT) applications, low latency has become a major concern for service providers to ensure user experience. Traditionally, the services can be supported via cloud computing platforms, however, due to the increase in network loads, accessing resources from remote servers in Clouds can lead to higher transmission costs and delays, which can also degrade the user experience~\cite{XuCSUR2019}. To improve the user experience, computing resources can be moved in close proximity to the end-users. Therefore, as a new paradigm complementing cloud computing, Mobile Edge Computing (MEC) has been proposed to enable efficient management of applications with low latency and constrained energy~\cite{IEEEhowto:Galloway}.

The motivation of MEC is to extend the cloud computing capabilities to the edge of the network \cite{XU2019JSS}. The services and applications can be deployed on the edge servers, and the part of the tasks on users’ devices can be offloaded to the edge servers.
%Revised by Qiheng:
{Therefore, services can be executed coordinately based on the resources provisioned by both local mobile devices and edge servers \cite{Wu2020IoTJ}.} 
This feature has made MEC an attractive paradigm for many delay-sensitive applications, e.g., Internet of Vehicles (IoV), Augmented Reality (AR), and Virtual Reality (VR)~\cite{Antonio2019SPE}. 

However, how to deploy the services to the specific edge servers is a challenge, since deploying the service to an edge server with a long distance to the user's device can still cause high latency and the edge servers can be heterogeneous \cite{Ali2020SPE}. Failing to support the users with satisfactory Quality of Service (QoS) will undermine the benefits of the MEC paradigm. In addition, compared with the traditional cloud computing paradigm, %deploying edge services to edge servers becomes a more challenge . 
%Currently, the applications can be constructed as a set of self-contained components, which are also called microservices. The microservices bring flexibility and scalability as they can be deployed and migrated in a fine-grained manner. Compared with the virtual machines deployment and migration, the microservices are much more lightweight. Benefiting from these features, many enterprises, such as Facebook and Amazon, have started to apply the microservice-based architecture. Some service providers, like Google, have also established container management platform (Kubernetes) derived from microservice paradigm. 
MEC brings a new challenge when mobile users (e.g., cars) move among different geographical locations~\cite{Guo2019SPE}. The distance and the latency between the users and edge devices keep changing dynamically due to the movement of users. Therefore, the services are required to be redeployed or migrated to suitable locations for the users to ensure QoS.

To address the above challenges, some key questions should be answered, including how to place the edge services in the initial phase as well as when and where to migrate the edge services according to the system running status. Trade-offs can exist when selecting the moment and locations to migrate services. For example, too frequent migrations can lead to communication overheads, while not migrating services can produce service response delay when a user is moving from the original location to another location that is far away from the original one.
%Added by Qiheng:
The running and network status of edge servers change with time, which may also affect the service delay.
When it comes to developing efficient service allocation algorithms for the MEC environment, it is significant to take the delay and user mobility into consideration. However, an optimal deployment solution for edge services can become inefficient if the deployment takes a few minutes to complete. Therefore, an efficient service allocation algorithm that can be executed within a short time would support the MEC environment in a better manner \cite{Badri2020TPDS}.

In this work, we use the probabilistic-based approach to study the service allocation in MEC. %We provide novel contributions beyond \cite{Wan2019}. 
%Compared with the Markov Decision Process (MDP) approaches, our approach can achieve near performance in reducing overall cost while the algorithm complexity is much lower.
Benefiting from the low algorithm complexity, our approach can be easily extended to the realistic environment. The main \textbf{contributions} are as below:

\begin{itemize}
\color{black}    \item We formulated the optimization problem as a service allocation decision problem to minimize the overall delay and cost, thereby reducing the service delay and migration costs. User mobility related to service delay is also considered in our problem. We prove that the optimization problem is NP-hardness.
\color{black}

    \item We proposed a delay-aware and mobility-aware approach based on Bernoulli trial to determine the decision of deploying services. Through the theoretic analysis, we prove that our approach can be bounded with $1 + \frac{(2 + \varepsilon + \delta) JR}{(1+ \varepsilon + \delta)(J+R)} $. 
    \item We conducted experiments in iFogSim \cite{gupta2017ifogsim} with the base station dataset and the taxi movement dataset derived from realistic traces \cite{crawdad} to evaluate the performance of our proposed algorithm. The results demonstrate that our approach can achieve better performance compared with baselines in terms of overall delay and costs with significant reduction, e.g., 8\% to 20\% reduction in overall delay and over 75\% decrease in migration cost. 
\end{itemize}

The remainder of this paper is organized as follows: Section 2 discusses and compares the related work. Sections 3 presents the system model and the target problem. Section 4 introduces our proposed approach for edge services management. Section 5 demonstrates the results based on realistic datasets by simulating the taxi movement scenario in iFogSim. Finally, the conclusions and future research directions are provided in Section 6.

\section{Related Work}

Edge service allocation management in MEC has been investigated by some research work. The related work can be categorized with 1) delay-aware management and 2) mobility-aware management. 

\begin{table*}[]
	\centering
	\caption{Comparison of related work}
	\label{tab:relatedwork}
	\resizebox{0.75\textwidth}{!}{%
		\begin{tabular}{|l|c|c|c|c|c|c|c|c|}
			\hline
			\multirow{2}{*}{\textbf{Approach}} & \multicolumn{2}{c|}{\textbf{Algorithm Complexity}} & \multicolumn{2}{c|}{\textbf{User Mobility}} & \multicolumn{2}{c|}{\textbf{Service Management}} & \multicolumn{2}{c|}{\textbf{Service Amount}} \\ \cline{2-9} 
			& \textbf{High} & \textbf{Low} & \textbf{Yes} & \textbf{No} & \textbf{Initial Placement} & \textbf{Migration} & \textbf{Single} & \textbf{Multiple} \\ \hline
			Wang \textit{et al.}~\cite{WangTMC} & $\surd$ &  & $\surd$ &  &  & $\surd$ &  & $\surd$ \\ \hline
			Wang \textit{et al.}~\cite{Wang2019ToN} & $\surd$ &  & $\surd$ &  &  & $\surd$ & $\surd$ &  \\ \hline
			Wang \textit{et al.}~\cite{Wang2015} & $\surd$ &  & $\surd$ &  &  & $\surd$ & $\surd$ &  \\ \hline
			Samanta \textit{et al.}~\cite{Samanta2020} &  & $\surd$ &  & $\surd$ &  & $\surd$ & $\surd$ &  \\ \hline
			Samanta \textit{et al.}~\cite{Samanta2019} &  & $\surd$ &  & $\surd$ &  & $\surd$ & $\surd$ &  \\ \hline
			Poularakis \textit{et al.}~\cite{Poularakis2019InfoCom} &  & $\surd$ &  & $\surd$ & $\surd$ &  &  & $\surd$  \\ \hline
			Pasteris \textit{et al.}~\cite{Pasteris2019InfoCom} & $\surd$ &  &  & $\surd$ & $\surd$ &  & $\surd$ &  \\ \hline
			Zhang \textit{et al.}~\cite{ZHANG2019111} & $\surd$ &  &  & $\surd$ & $\surd$ &  & $\surd$ &  \\ \hline
			Wan \textit{et al.}~\cite{Wan2019} & $\surd$ &  &  & $\surd$ & $\surd$ &  & $\surd$ &  \\ \hline
			Gao \textit{et al.}~\cite{Gao2019} &  & $\surd$ &  & $\surd$ & $\surd$ &  & $\surd$ &  \\ \hline
			Yu \textit{et al.}~\cite{Yu2018GlobeCom} &  & $\surd$ &  & $\surd$ & $\surd$ &  & $\surd$ &  \\ \hline
			Badri \textit{et al.}~\cite{Badri2020TPDS} &  & $\surd$ &  & $\surd$ & $\surd$ &  & $\surd$ &  \\ \hline
			Ouyang \textit{et al.}~\cite{Ouyang2019INfoCOm} &  & $\surd$ & $\surd$ &  &  & $\surd$ & $\surd$ &  \\ \hline
			Wu \textit{et al.}~\cite{Wu2019ICWS} & $\surd$ &  & $\surd$ &  & $\surd$ &  & $\surd$ &  \\ \hline
			
			{\color{black}Ghosh \textit{et al.}~\cite{ghosh2019mobi}} & {\color{black}$\surd$} &  & {\color{black}$\surd$} &  &  &  &  & {\color{black}$\surd$} \\ \hline
			
			{\color{black}Shi \textit{et al.}~\cite{shi2017maga}} & {\color{black}$\surd$} &  & {\color{black}$\surd$} &  &  &  &  & {\color{black}$\surd$} \\ \hline
			{\color{black}	Yu \textit{et al.}~\cite{YU2018722}} & {\color{black}$\surd$} &  & {\color{black}$\surd$} &  &  &  &  & {\color{black}$\surd$} \\ \hline
		
			Our approach &  & $\surd$ & $\surd$ &  & $\surd$ & $\surd$ &  & $\surd$ \\ \hline
		\end{tabular}%
	}
\end{table*}

\subsection{Delay-Aware Edge Service  Management}
Several articles based on Markov Decision Process (MDP) have been proposed. Wang \textit{et al.}~\cite{WangTMC} investigated service coordination among edge clouds to support delay-aware service migration for mobile users. A reinforcement learning approach based on the MDP model is also applied to find the optimal decisions on migrating services. Wang \textit{et al.}~\cite{Wang2019ToN} presented a dynamic service migration approach for MEC based on MDP to find the optimal solution.  An approach based on service migration by using MDP to capture costs to design service migration policy was also proposed in~\cite{Wang2015}.
However, the solution space of these approaches grow fast when the number of edge devices and users increase, the scalability can be the limitation for these approaches to be applied to the realistic scenario. Therefore, some assumptions need to be relaxed when applying in a practical scenario. Different from these works, our proposed work can be performed with much lower complexity. 

Samanta \textit{et al.}~\cite{Samanta2020} studied a dynamic microservice scheduling scheme for MEC-enabled IoT to improve the network throughput and ensure the QoS for users. A distributed and latency-aware microservice scheduling policy was also introduced to reduce service latency while ensuring transmission rate to minimize microservices completion time~\cite{Samanta2019}. Poularakis \textit{et al.}~\cite{Poularakis2019InfoCom} investigated an algorithm to jointly optimize service placement and request routing to support multi-cell networks in MEC under multiple resource constraints. The proposed algorithm can achieve near-optimal performance for utilizing available resources to maximize the number of requests and reducing latency. Pasteris \textit{et al.}~\cite{Pasteris2019InfoCom} considered service placement in MEC with heterogeneous resources. Their objective is to maximize the total reward by using an approximate algorithm. The designed algorithm can be applied to both small scale and large scale environments with two subroutines. 
Unfortunately, user mobility is not considered in these works, and most of them focused on the initial placement of services. In contrast, both the initial service placement and dynamic service migration are considered in our work.

 \subsection{Mobility-Aware Edge Service Management}

Mobility-aware service management in MEC has attracted the attention of researchers. Zhang \textit{et al.}~\cite{ZHANG2019111} presented a deep Q-network based approach for task migration in the MEC system, which can obtain the optimal migration policy without knowing the information of user mobility. Wan \textit{et al.}~\cite{Wan2019} introduced a joint optimization methodology to assign resources to tasks based on evolutionary computation by considering power usage and latency in MEC environment. Gao \textit{et al.}~\cite{Gao2019} introduced an approach to jointly optimize the access network selection and service placement in MEC. The long-term optimization problem is decomposed into a set of one-shot problems to reduce computation time. Yu \textit{et al.}~\cite{Yu2018GlobeCom} investigated collaborative service placement in MEC to relieve the limited capacity of base stations by collaboratively utilize the resources from adjacent base stations. To solve this optimization problem, a decentralized algorithm based on matching theory was also proposed. Badri \textit{et al.}~\cite{Badri2020TPDS} considered energy-aware service placement in MEC as a multi-stage stochastic program, which can maximize the QoS under the constraint of limited energy. A parallel sample average approximate algorithm was also proposed to solve the energy-aware problem. 
These works consider multiple objectives optimization, while the service migration is not modelled in these articles. In contrast, we consider service migration in our proposed approach. 

Ouyang \textit{et al.}~\cite{Ouyang2018} proposed a mobility-aware service placement framework for cost-efficient MEC. In this approach, the latency is reduced under the constraint of long-term migration cost based on Lyapunov optimization. 
An adaptive service placement mechanism aimed at improving latency and service migration cost was also presented~\cite{Ouyang2019INfoCOm}, which was modelled as a contextual multi-armed bandit problem and solved by an online algorithm based on the Thompson-sampling approach. 
Wu \textit{et al.}~\cite{Wu2019ICWS} formulated the mobility-aware service selection problem in MEC as an optimization problem and solved it by a heuristic algorithm that combines the genetic algorithm and simulated annealing algorithm. The proposed approach can reduce the response time of service and the algorithm execution time is one order of magnitude lower than the baselines. Our work differs from these studies by considering multiple services management rather than single service management. 

{\color{black} 
 Ghosh \textit{et al.}~\cite{ghosh2019mobi} proposed a mobility-driven cloud-fog-edge collaborative real-time framework. The framework considers user mobility and performs location prediction of users using the Hidden Markov model. It enables efficient information processing and reduces delay. Although considering user mobility, these approaches focus on computation offloading, in mobile edge computing.
Shi \textit{et al.}~\cite{shi2017maga} introduced a mobility-aware computation offloading decision method. It takes user mobility into consideration and adopts an adaptive genetic algorithm for offloading decisions. The proposed method can achieve a better offloading success rate and lower energy consumption of mobile users. Yu \textit{et al.}~\cite{YU2018722} presented a dynamic mobility-aware partial offloading algorithm to investigate the amount of data to be offloaded, which is based on location prediction to minimize energy consumption and service delay. Different from the service management approach in our approach, they perform location prediction and optimize the service delay via offloading decision instead of service placement among edge servers.
}

\subsection{Critical Analysis}
Our proposed work and the related work is compared in Table~\ref{tab:relatedwork}. As a summary, our work contributes to the growing body of research in service allocation in MEC. To address the delay-aware and mobility-aware challenges of MEC, we apply a probabilistic-based algorithm for both service initial assignment and dynamic service migration with low computation complexity. We also consider the user mobility and multiple service management during our service allocation process. In addition, the performance of our approach is validated based on the data derived from realistic traces.

\section{System Model and Problem Statement}

\color{black}In this section, the system model is introduced for the key components of service scheduling in the typical MEC system. The model provides mechanisms for abstracting various functions and operations into an optimization problem. A case study is also provided to clarify the process of service scheduling in the MEC scenario. \color{black}

\color{black} \subsection{MEC System Model}

\begin{figure}[ht]
	\centering
	\includegraphics[width=0.6\textwidth]{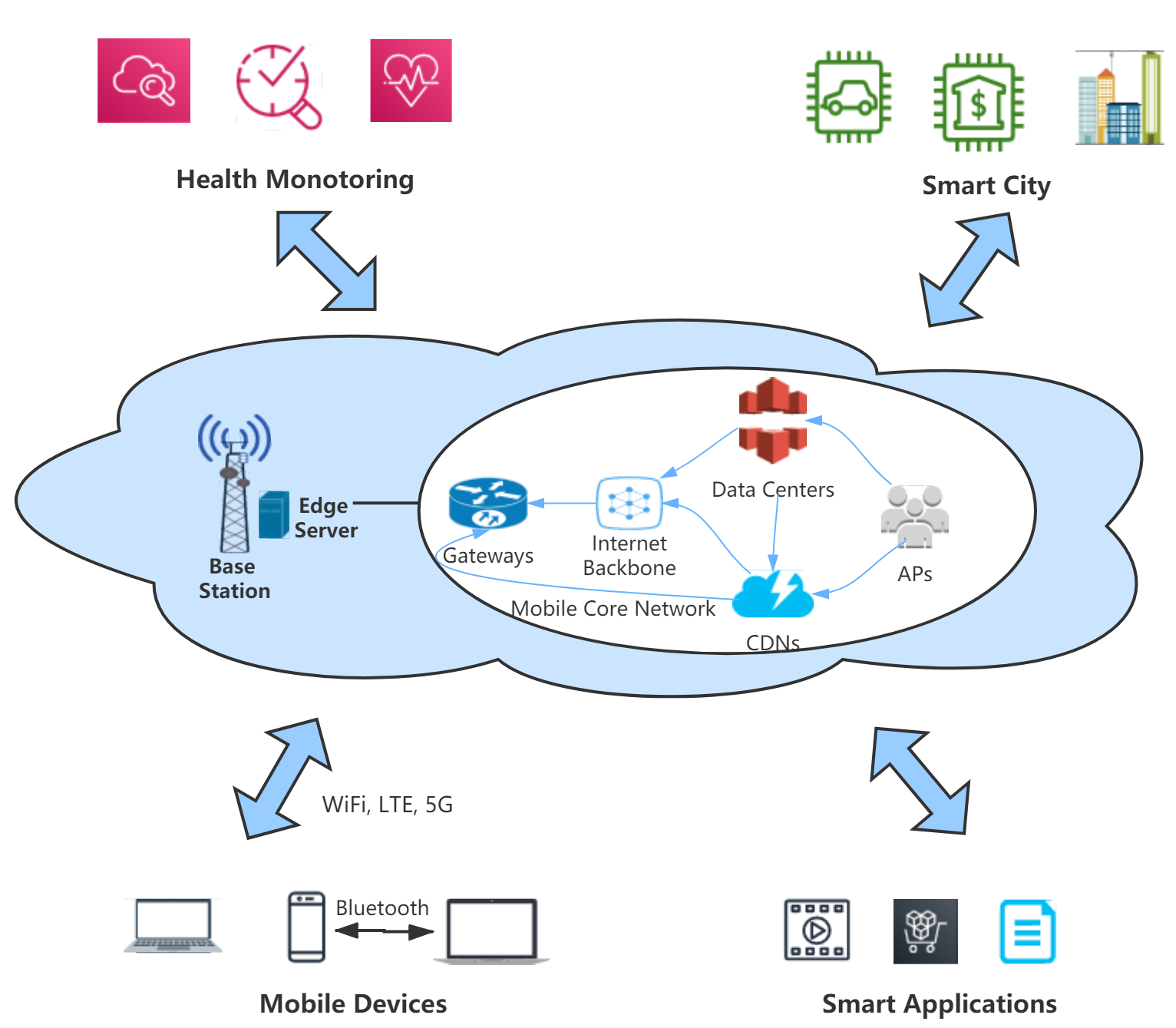}
	\caption{\color{black}MEC System Model}
 	\label{fig:MECModel}
\end{figure}

As shown in Figure~\ref{fig:MECModel}, the typical MEC system model can provide support for various applications, e.g., health monitoring, smart city, and smart mobile applications. The model contains several key components include mobile users and edge servers. Mobile users can utilize mobile devices and applications to request services from the edge servers via wireless access points (APs). The edge servers are generally small-scale data centers deployed by cloud and telecom operators near mobile users, which can connect to data centers through a gateway via the Internet. The edge servers and mobile users are separated by the air interface based on the advanced wireless communication and networking technologies. 

From the communication perspective, in MEC systems, communications are typically between mobile users and APs with the possibility of the device to device (D2D) communications. The edge servers can be co-located with wireless APs, such as base stations and WiFi routers, which can significantly reduce the capital expenditure. Apart from providing the wireless interface for edge servers, the wireless APs can also support access to resources from remote data centers via backhaul links. The computation tasks can be offloaded between the edge serves and remote data centers. If the mobile device cannot connect to edge servers due to the limited wireless interfaces, the D2D communications with neighboring devices can be complementary. To improve user experience, the content delivery networks (CDNs) provide the cached data that enable the users to access the data efficiently. 

Currently, different commercial technologies can be utilized for mobile communications, e.g., 5G network based on the combination of long-term evolution (LTE) and new radio-access technologies has been standardized and put into commercial use. These technologies can support efficient wireless communication from mobile users to APs for varying data rates and transmission ranges. Bluetooth can be used for short-range D2D communications in the MEC system. And WiFi, LTE, and 5G are more powerful technologies for long-range communications between mobile users and edge servers, which can be dynamically switched based on the link reliability. \color{black}

\subsection{Problem Definitions}
In this subsection, we will introduce the objectives that our approach aims to achieve, including the overall delay and migration cost, which present the costs from the perspectives of user and service provider respectively.

\subsubsection{Basic Entities}
\color{black}
In our model, we use $U$ to represent a set of users with size $N$, where $u_i \in U$ and $i \in \{1, 2, \ldots, N\}$. 
Let $E$ to represent a set of edge servers with size $J$, where $E_j \in E$ and $j \in \{1, 2, \ldots, J\}$
and $BS$ to represent base stations with size $L$, $BS_l \in BS$ and $l \in \{1, 2, \ldots, L\}$. Let $S$ denote the set of services with size $R$, where $S_r \in S$ and $r \in \{1,2, \ldots, R\}$.

%Revised by Qiheng: la->lat, lg->lng
For each user $u_i \in U$, we denote geographic location at time $t$ by $p_i(t)=(lat_i, lng_i)$, specifically in latitude $lat_i$ and longitude $lng_i$. Each user utilizes one edge service represented as $S_{u_i}$, which needs to be deployed on edge servers. 
We denote the edge server by $E_j \in E$, which is attached to a base station $BS_l$. Therefore, the geographical location of the edge server is the same as that of the base station, denoted as $p_l=(lat_l, lng_l)$. The resource capacity of each edge server is represented by $C_{E_j}(a_j, m_j, n_j)$, where $a_j, m_j$, and $n_j$ represent the CPU, memory, and network capacity respectively, and the capacity will affect the computation latency of the services. 
User $u_i$ accesses the services $S_{u_i}$ from the edge servers $E_j$ is denoted as $E_{u_i,j}$. 
%Each edge server can host multiple services. We use $DS(i)$ to denote all services currently deployed on the edge server $ES_i$.

At time $t$, user $u_i$ connects to the nearest base station $BS_l$, we denote the distance between the user $u_i$ and the base station $BS_l$ as $D_{i,l} = \left \| p_i(t)-p_l \right\|$, that is the Euclidean distance between the user and the base station. The nearest base station is denoted as $current\ base\ station$, and as the user moves, $current\ base\ station$ of the user will automatically switch. 
And, $selected\ base\ station$ denotes the base station on which the edge server that hosts $u_i$'s service is deployed. 
Because of the switch of $current\ base\ station$ and the variation of the edge server's workload, the latency of the edge service may increase, which are modelled as parameters in our model. Therefore, a scheduling algorithm is required to determine if the service should be migrated to another edge server and to choose a destination server to perform the service migration. 

\color{black}

\subsubsection{Overall Delay}
In our model, the $overall\ delay$ mainly consists of three parts, namely, $communication\ delay$, $computation\ delay$ and $migration\ delay$. We will explain these three parts of delay in the following sections, which can determine the user experience.
%Revised by Qiheng: which can represent the user experience. 

\textbf{Communication delay} can be split into two parts, including the delay of the data transmission from user $u_i$ to its $current\ base\ station$ and the delay of $current\ base\ station$ forwarding data to the base station that hosts the edge server. 

Users perform the first part through the wireless channel. Here, we apply Eq.~(\ref{eq:transRate}) for calculating the maximum transmission rate ($tr$) of the wireless channel based on Shannon Theory~\cite{wyner1974recent}.
\begin{equation}
   tr = W\log_{2}{(1+\frac{S_p}{gN_p})},
   \label{eq:transRate}
\end{equation}
where $W$ denotes the channel bandwidth, $S_p$ denotes the transmission power of the mobile device and $N_p$ denotes the noise power. Besides, channel gain between the location of the user and its $current\ base\ station$ is denoted as $g$, varying as the user moves from one place to another.

The second part of the $communication\ delay$ is the transmission delay between $current\ base\ station$ ($BS_{c}$) and $migrated\ base\ station$ ($BS_{m}$). We use a matrix $M_{c,m}$ to represent the delay between $BS_c$ and $BS_m$. $M_{c,m}$ is infinite if $BS_c$ and $BS_m$ are not connected directly. Therefore, the second part of the communication delay can be computed by finding the shortest path with the minimum transmission delay, which we denote as $D(BS_{c}, BS_{m})$.
%TODO: migrated -> selected

%Providing the task size of $u_i$ is $c_i$, 
Then, we can get the communication delay $T_{cm}$ by:
\begin{equation}
    T_{cm}(u_i, E_j, E_m) = \frac{c_i}{tr} + D(BS_{c}, BS_{m}).
   \label{eq:communication}
\end{equation}
where $c_i$ is the task size of $u_i$.
%  Therefore, the $communication\ delay$ can be represented as
 
\textbf{Computation delay} is the task execution time of services deployed on the edge server. Since each edge server can host several services and execute multiple tasks at the same time, the execution time of each task varies 
%Revised by Qiheng: due to the resource utilization of the edge server.
due to the available resources of the edge server.
%Suppose that the task size is of $u_i$ is $c_i$, %the computational intensity of the task is $w_t$ 
%and the computational intensity allocated to the task by edge server $E_j$ is $w_j$, 
The task execution time $T_{cp}$ can be calculated by:
\begin{equation}
    T_{cp}(u_i, E_j) = \frac{c_i}{w_j}.
   \label{eq:execution}
\end{equation}
where $w_j$ is the computational workload allocated to the task by edge server $E_j$.

\textbf{Migration delay} is the downtime of service migration. During the migration, the service needs to be suspended for a period of time and then the in-memory state of the service will be transferred to the destination edge server. Then, the service will restart on the new edge server and process service requests from mobile users. Therefore, the migration will also increase the service delay due to the downtime, which should also be considered in the overall delay. 
The migration delay $T_m$ can be modelled as: 
\begin{displaymath}
T_m(BS_c, BS_m)  = \left\{ \begin{array}{ll}
0, & \textrm{if $BS_c = BS_m$},\\
M_c, & \textrm{if $BS_c \neq BS_m$},
\end{array} \right.
\end{displaymath}
when the migration is not triggered, the migration delay is 0, and $M_c$ can be a constant, indicating the migration delay. 

Therefore, the overall delay can be calculated based $T_{cp}$ and $T_{cm}$. One of our objectives is to minimize the overall delay to assure user experience that can be represented as:
\begin{equation}
\label{eq:overallDelay}
    \min:\quad\frac{1}{P\times N \times J} \sum_t^P \sum_i^N \sum_j^J \big\{T_{cp}(u_i, E_j) + T_{cm}(u_i, E_j, E_m) + T_m(BS_c, BS_m)\big\},
\end{equation}
where $t\in \{1,2,\ldots,P\}$ represents the observed time interval.

\subsubsection{Migration Cost}
The migration cost is applied as another important metric in our model. From the perspective of service providers, it represents the cost for migrating services and the placement of services. The migration cost of service $S_r$ is denoted as $C^{S_r}_{j,m} = F(S_r, E_j, E_m)$, 
%Revised by Qiheng: where  $F()$ is the function to calculate the computation delay and communication delay from Equations (2) and (3). 
where $F$ is the function to calculate the transmission cost and computation cost of service migration.

As a user keeps moving in a period of time, its service may be migrated many times to reduce the delay and ensure the quality of experience. Thus, we can compute the sum of the cost of all service migration and obtain the migration cost by Eq. (5): \begin{align}
  \min:&\quad C_{o} = \sum_{r=1}^R{\sum_{k=1}^{|SM^r|}}C^{r}_{SM^{r}_{k,j}, SM^{r}_{k,m}}, \\
   \textrm{s.t.}:& \quad\sum_{i} c_i \leq c_j, \forall E_{u_i,j},
   \label{eq:migrationCost}
\end{align}
where $SM^r$ is the set of all services to be migrated, $SM^{r}_{k, j}$ and $SM^{r}_{k, m}$ denote the source edge server $E_j$ and the destination edge server $E_m$ of the $k^{th}$ migration of $SM_r$, respectively. %The requested resources of services should be no more than the maximum capacity of $E_j$.

Our objective function is to minimize the migration cost $C_o$ in Eq. (5), while satisfying the constraint in Eq. (6) that the requested resources of services should be no more than the maximum capacity of $E_j$.

%\begin{equation}
%    s.t. \sum_{i} c_i < c_j, \forall E_{u_i,j}
%\end{equation}

\color{black}\subsection{Case Study} 

\color{black}

\begin{figure}[ht]
	\centering
	\includegraphics[width=0.75\textwidth]{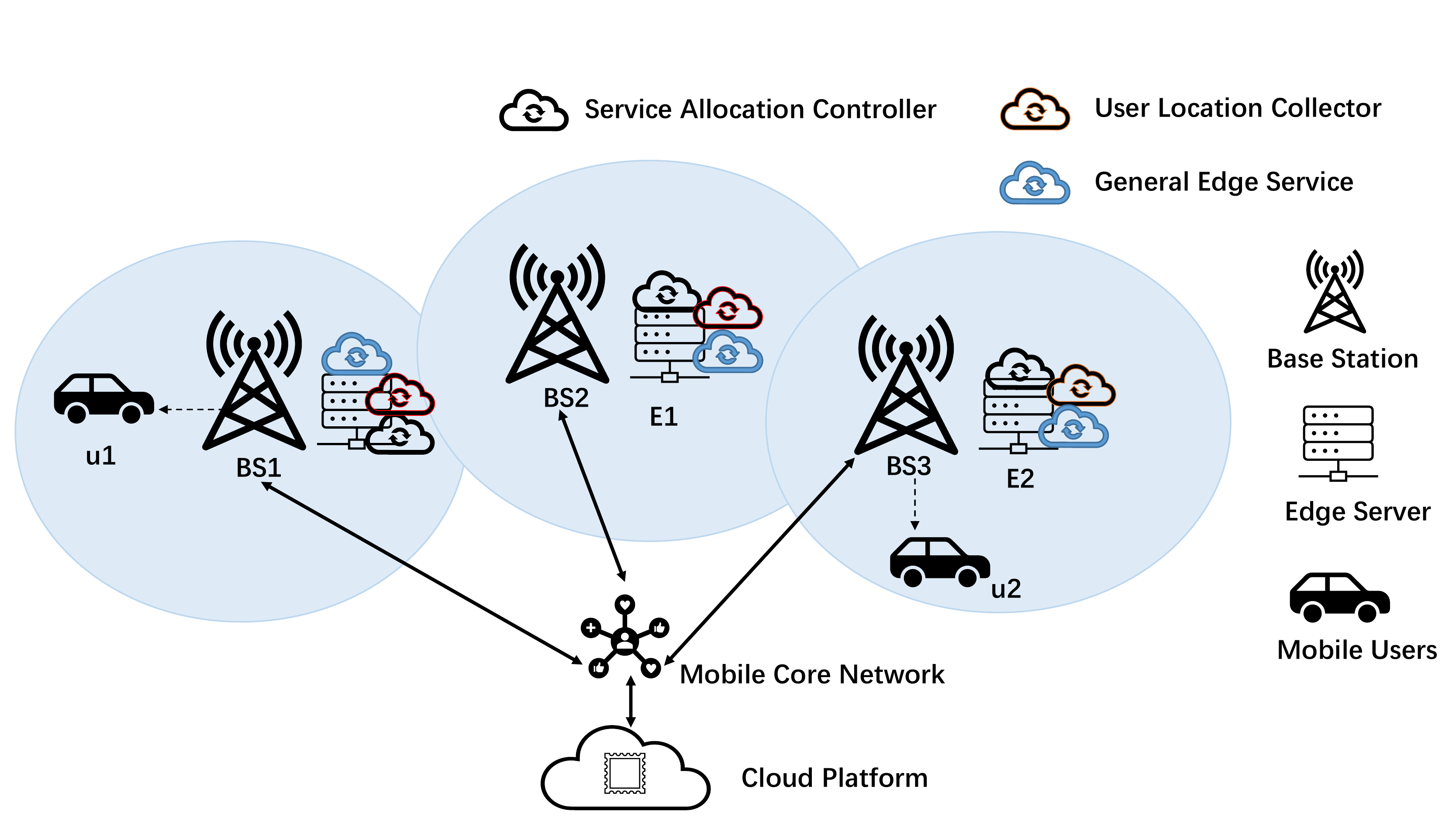}
	\caption{\color{black}Case Study}
 	\label{fig:systemModel}
\end{figure}
\color{black}

As shown in Figure~\ref{fig:systemModel}, we consider a case study consisting of a set of base stations, each co-located with an edge server deployed with multiple edge services, and a group of mobile users moving around among different areas covered by different base stations. To handle the service scheduling process and ensure service continuity, the entities in this case study include: (1) a MEC platform containing edge servers, (2) edge services execute users’ requests, (3) user location collector, (4) service migration controller, and (5) a virtualization infrastructure.

Each mobile user is associated with a specific edge service, which handles users’ requests and can be migrated to another edge server by tracking the mobile user’s movements. For example, as a mobile user approaches the edges of a base station that can be covered, the user location collector running on edge servers informs the nearby edge servers that the mobile user is about to perform the handover to a new area covered by another base station. The information is then used by the service allocation controller to decide whether to migrate the edge service to the other edge server, and in that case, which edge server is to be performed the migration. The service allocation controller has an overview of the entire MEC system and server as the orchestration. Finally, the virtualization infrastructure provides computation, storage, and network resources to provision resources for edge services. It can also manage the migration process by collecting information about the remaining resources of edge servers.

In this use case, as shown in Figure~\ref{fig:systemModel}, we consider the following situation: first, user $u_1$ connects to the base station $BS_1$, but the service required by $u_1$ is deployed on $E_1$ built aside $BS_2$. Thus, to access the service deployed on $E_1$, $u_1$ should transmit message to $BS_1$ first, and then $BS_1$ forwards the message to $BS_2$, which sends the packet to $E_1$. On the other hand, user $u_2$ connects to the base station $BS_3$, attached by the edge server $E_2$ hosting the service of $u_2$, so the message transmission between $u_2$ and $E_2$ does not need other base stations. Our objective is to optimize the allocation of edge services on edge servers to avoid the high delay for users.

Although the use case can be applied to both live and static service migrations, we focus on live and stateful service migration processes. Service migration is ensured between edge servers through the backhaul links that have sufficient data rate, thus the performance of service migration performance is not undermined by network traffics. This assumption can be relaxed, but in this paper, we restrict the analysis to this case, which also conforms with the powerful capability of the 5G scenario. 

To realize our proposed approach, at each time interval, mobile user locations should be gathered to calculate the delay between the user and edge servers. The utilization of edge servers is given as the probability to perform the service allocation between edge servers. The probability is then used in the service scheduling algorithm to compute the possibility of edge servers accepting migrated services. To be noted, the mobile user location collector and the scheduling algorithm in the service allocation allocator are independently executed. Therefore, when the service allocation controller decides which edge service to be migrated and where should be migrated, the mobile user location collector sends the most updated location data for each mobile user.  

\color{black}

\subsection{Proof of NP-hardness}
In this subsection, we prove that the problem we aim to solve is an NP-hardness problem, and the proof is as below:

\textbf{Proof:} We consider the decision version of the set cover problem \cite{Pasteris2019InfoCom}, which is NP-complete. We have set $\mathcal{Z}$, a set $\mathcal{A}$ of subset $\mathcal{Z}$, and a number $k \in \mathbb{N}$. The objective is to find if a set $\mathcal{B}$, which is subset of $\mathcal{A}$, can exist to satisfy $|\mathcal{B}| \leq k$ and $ \bigcup\mathcal{B}= \mathcal{Z}$.

We define the Service Migration with Set Constraints (SMSC) problem as follows:
We assume that we have $|\mathcal{A}|-k+1$ types of services. And each type of service can have multiple services. Each type of service will only be deployed on a single node, which means one node will not have more than one service with the same type. One of the service types is special and denoted as $i'$. Let $S'\triangleq S\backslash \{i'\}$. We also define $V\triangleq \mathcal{A}$. Every node is a subset of $\mathcal{Z}$. For every normal service type $i$, we have a user $u_i$ that needs to connect with this type of service.  For this user, we define $S_{u_i} \triangleq i$ and $E_{u_i} \triangleq V$. For every $z \in \mathcal{Z}$, there is a user $u_z^{'}$. For the user, let $S_{u_z^{'}} \triangleq i'$ and $E_{u_z^{'}} \triangleq \{Y\in\mathcal{A}:z\in Y\}$.

We consider that the solution to the set cover problem is $\mathcal{X}$, and the service migration solution is defined as $M$. For each service $i \in S'$, it chooses an node $j_i$ in $\mathcal{A} \backslash \mathcal{X}$ and $j_i \neq j_{i^*}$. For $i^* \in S'$ with $i^* \neq i$. This can be assured as $|\mathcal{A} \backslash \mathcal{X}| \geq |\mathcal{A}| - k = |S'|$. For each service $i \in S'$, we can define $M_i \triangleq {j_i}$ and $M_{i'} = \mathcal{X}$. This migration operation is feasible as a single node is only deployed with one instance of the same type of service. 
The objective that every user can be satisfied with QoS can be denoted as: consider a service $i \in S'$, the user $u_i$ can be satisfied if $M_{S_{u_i}} \bigcap E_{u_i} = M_i \bigcap V = M_i \neq \emptyset$. For every $z \in \mathcal{Z}$, the user $u_z^{'}$ can be satisfied if there is a set $Y' \in \mathcal{X}$ with $z \in Y$, therefore, $M_{S_{u_i^{'}}}\bigcap E_{k_z^{'}} = M_{i^{'}} \bigcap \{Y \in \mathcal{A}: z \in Y\} = M \bigcap \{Y \in \mathcal{A}: z \in Y\} \supseteq  \{Y'\}$. The above proof shows that a solution exists to satisfy the QoS of users in SMSC problem. 

If all the users are satisfied as defined in the above SMSC problem, we have for every $i \in S'$, the user $u_i$ is satisfied with the QoS requirement. Therefore service $S_{u_i}=i$ must be placed on some nodes. Let $\mathcal{K}$ be the set of all edge nodes $j$ that services in $S'$ are deployed. Since only one service of the same type will be placed on the same node,  we can have $|\mathcal{K}| \geq |S'|$, which is bounded below by $|\mathcal{A}| - k = |V| - k$. Define $\mathcal{A} = V\backslash\mathcal{K}$ which has cardinality at most $k$. 
For any $z \in \mathcal{Z}$, if user $u_z^{'}$ is satisfied, there must be at least one node in $E_{u_z^{'}} = \{Y\in\mathcal{A}:z\in Y\}$ that service $ S_{u_z^{'}} = i'$ is placed. We denote this node as $Y'$ and $Y' \notin \mathcal{K}$. The reason is $i'$ is already deployed there. We consider that a single node only hosts a single instance of the same type of service, no need to deploy services on nodes in $S'$. Then we can have $Y'\in \mathcal{X}$ and $z \in \bigcup \mathcal{X}$. Since it holds every $z\in \mathcal{Z}$, we can have $\bigcup\mathcal{X}=\mathcal{Z}$, which means $\mathcal{X}$ is the solution to the set cover problem.

The above proves that the SMSC is NP-hard.

\section{Probabilistic Delay-aware and Mobility-aware Approach for edge service management}
In this section, we introduce \textit{\textbf{P}}robabilistic based \textit{\textbf{D}}elay-aware and \textit{\textbf{M}}obility-aware \textit{\textbf{A}}pproach (\textit{\textbf{PDMA}}) for edge service management.
We focus on two service scheduling procedures, \textit{Service Assignment} and \textit{Service Migration}. We present our probabilistic-based algorithms to perform service scheduling while reducing service latency and migration costs. The proposed approach is inspired by the probabilistic method proposed in [26] for VM consolidation in the cloud computing environment, and revisions have been made to adapt to the scenario of MEC.

\subsection{Service Assignment}
In the service assignment procedure, when a mobile user sends a service request, the service assignment algorithm assigns a suitable edge server to host the service. In most service assignment algorithms, cloud coordinators need to perform computation-intensive calculations to determine an allocation decision, which requires massive computing resources and long processing time. Additionally, in MEC, the edge-to-cloud communication delay is much higher than the edge-to-edge delay. Therefore, performing the scheduling algorithms by the cloud coordinator will greatly increase the service delay.

To address this, our approach leaves the decisions to each edge server. The mobile user sends the request for a service assignment to the edge servers. And, instead of sending this request to the cloud, each edge server decides whether to host new services based on the current resource utilization and network status.

For example, we consider the CPU utilization of edge servers as the metric when making a service assignment decision. If the CPU utilization is close to or even exceeds the utilization threshold, the edge server is very likely to be overloaded after deploying a new service, which will decrease the processing capability of the edge server and thus lead to higher service latency. To avoid performance degradation, the edge server should not accept the service assignment request. On the other hand, those idle servers can be shut down to reduce energy consumption by migrating services to other servers. Other edge servers with moderate CPU utilization will have a higher success rate when making assignment decisions. However, different from the VM scheduling in the cloud data center, the service scheduling in MEC involves the data transmission of mobile services over long distances via wireless channels. Therefore, edge servers should also take the transmission cost into consideration when making service assignment decisions.

In our approach, the edge server performs a Bernoulli trial to make the service assignment decision. A successful trial indicates that the server can host a new service. We utilize an \textit{assignment function} to determine the probability of a successful trial. The assignment function is defined as follows:
\begin{align}
    f(x,p,T) &= \frac{1}{M_p} x^p(T-x),\ 0 \leq x \leq 1,
    \label{eq:assignment_function}\\
    M_p &= \frac{p^p}{(p+1)^{p+1}}T^{p+1},
    \label{eq:Mp}
\end{align}
where $x$ is the utilization of a certain resource of the edge server, $T$ is the upper threshold of the utilization of this type of resource, e.g., CPU, and $p$ is the shape parameter that can adjust the probability distribution. If $x > T$, the value of the assignment is 0, which means rejecting the service assignment request. $M_p$ is a regularization parameter used to adjust the value of the assignment function $f$ to the maximum value of 1.
\begin{figure}[ht]
	\centering
	\includegraphics[width=0.46\textwidth]{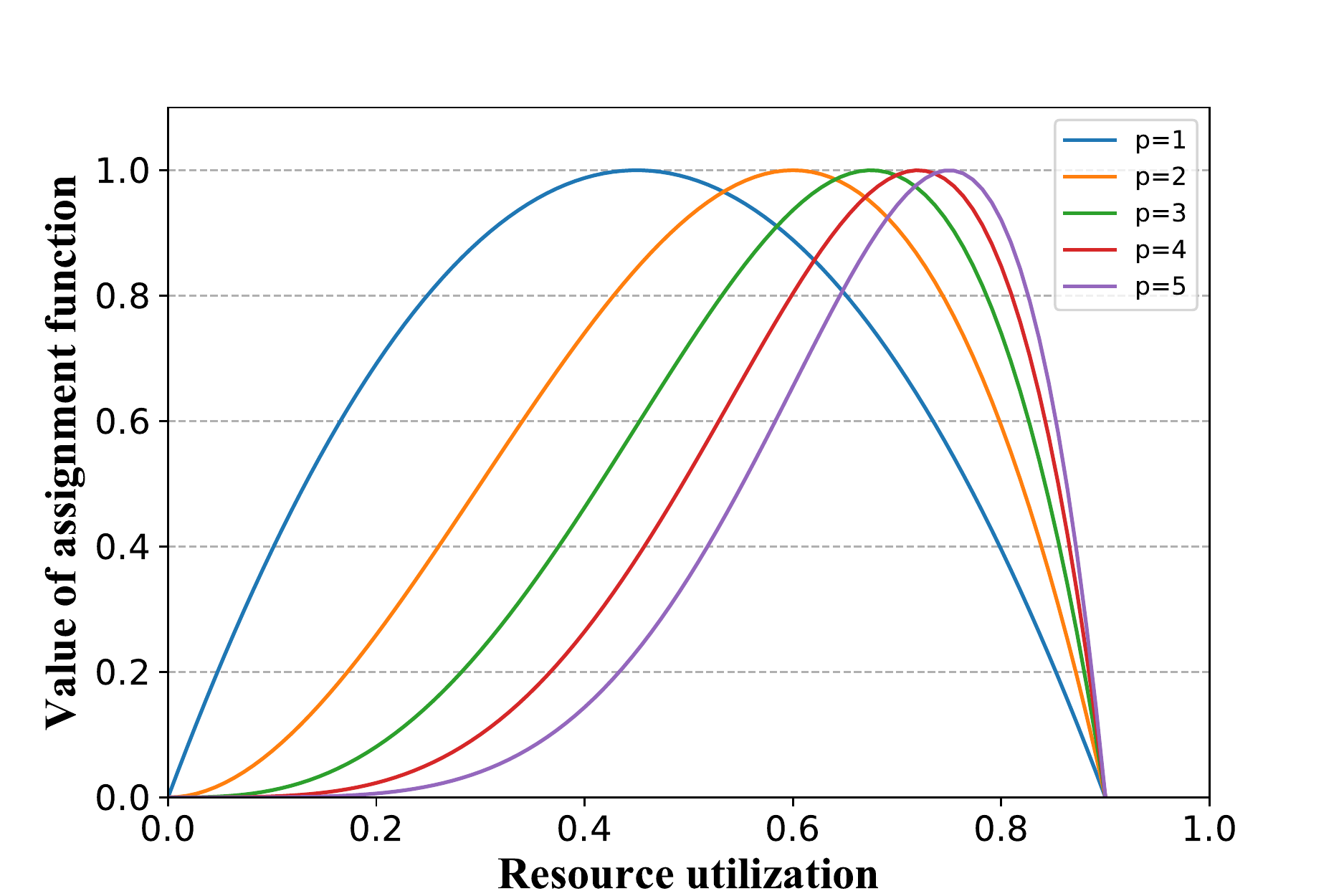}
	\caption{Assignment Function, $T$ = 0.9}
 	\label{fig:assignmentFunction}
\end{figure}

Figure \ref{fig:assignmentFunction} shows the distribution of the assignment function with varied $p$ value. Assuming that the resource utilization threshold $T$ is 0.9, the resource utilization $x$ changes within the range of [0,$T$]. As shown in the figure, under different shape parameters, the probability value of the assignment function first increases with the growth of resource utilization $x$, and reaches the maximum value when $x=\frac{pT}{p+1}$. This trend conforms to the basic idea of our service assignment. For different shape parameters, the value of this function is also very low when the resource utilization is close to the threshold, which avoids the allocation of new services to the server that is prone to be overloaded. We can also observe that when the shape parameter is larger, the highest acceptance probability is closer to the threshold. Therefore, we can alter the shape parameter based on the average resource utilization to adjust the probability of service assignment.

If the trial is successful, it means the edge server agrees to accept the deployment of the new service and responses with an acceptance message. The  coordinator is responsible for collecting all the messages and selecting the most suitable edge server to host the service. Specifically, the most suitable one can be the server with the least migration cost to minimize the impact of service migration on user service experience. If all the edge servers reject the assignment request, the coordinator should scale the edge data center (e.g., deploy a new edge server), and allocate the new edge services.

\color{black} The pseudocodes of the service assignment algorithm are described in Algorithm 1. The algorithm focuses on assigning the associated edge service for a mobile user to a specific edge server. The service assignment algorithm can be utilized for both initial service assignment or service assignment in the service migration process. 
First, the mobile user connects to the nearest base station and sends the service request to the edge server. If $initial$ is true, it means that the approach is performing the initial placement of a mobile service, so the algorithm sets the location of the service the same as the mobile user (lines 1-2). 
The algorithm attempts to collect assignment decisions from all edge servers that satisfy the delay threshold at the current time slot. $ES\_Candidate$ includes all edge servers that accept the service assignment request (lines 4-9). If the candidate list is not empty, we select an edge server that is nearest to the current location of service to assign the service.  Otherwise, if no available edge server is in the candidate list, the algorithm should perform $scaleUp$ to switch on an idle edge server to host the service (lines 10-14).

%The application of the low-complexity service assignment function avoids the  intensive computations. In addition, we can also extend the assignment function to include more factors (e.g. memory, storage) in our service assignment decision.

\begin{algorithm2e}[h]
\color{black}
    \caption{Service Assignment Algorithm}
    \KwIn{Mobile User $U_k$, Edge Service $S_k$, Edge Server List $ES\_List$, Initial Assignment $initial$, Delay Threshold $T_d$}

    \If{$initial == True$}{
        $S_k.location \leftarrow U_k.location$
    }

    $ES\_Candidate = []$

    \For{$ES_i\ in\ ES\_List$}{
        $delay \leftarrow getDelay(U_k, S_k, ES_i)$ according to Eq. (2)

        \If{$delay < T_d$}{
            $decision \leftarrow getAssignmentDecision(ES_i)$ based on Eq. (7)
            
            \If{$decision == Accept$}{
                % $ESCandidate.add(ES_i)$
                $ES\_Candidate = ES\_Candidate \cup ES_i$
            }
        }
    }

    \eIf{$ES\_Candidate\ is\ empty$}{
        $newES \leftarrow\ scaleUp()$
    }{
        $newES \leftarrow\ getNearest(S_k, ES\_Candidate)$
        % $SortByDistance(ESCandidate)$
    }
    % Assign $S$ to $newES$
    $Assign(S_k,newES)$
\end{algorithm2e}

\color{black}
\textbf{Complexity analysis of Algorithm 1:} the decision of whether the service is an initial assignment or not (lines 1-2) takes a time of $O(1)$; the candidate edge servers collection process (lines 4-9) takes a time of $O(J)$, where $J$ is the number of edge servers; the edge server selection from the candidate list (lines 10-14) takes a time of $O(Jlog(J))$, which is based on sorting algorithm; and the final assignment operation takes $O(1)$ time. Therefore, we can conclude that the time complexity of Algorithm 1 is $O(Jlog(J))$. 
	\color{black}

\subsection{Service Migration}
Since the mobile user moves in real-time, it may move away from the edge server, leading to higher communication latency and affecting the service experience. 
At the same time, the running status of mobile services can also change dynamically. For example, the load of edge servers may exceed the upper threshold and make the edge servers become over-utilized, resulting in performance degradation.
Over-utilized edge servers may not be able to process tasks efficiently, which leads to increase in the service delay. Therefore, it is required to detect the overloaded situation and perform dynamic service migration to optimize the deployment of services and reduce the delay. In the following part of this section, we will introduce our service migration algorithm.

\color{black}
We divide the service migration into two situations for consideration, and the pseudocodes are shown in Algorithm 2.

(1) \textit{Delay violates the threshold} (lines 3-9). 
%The communication delay between the mobile user and the edge server is monitored during the runtime of the service. 
The service delay is monitored during the runtime of the service. 
When the delay exceeds the predefined threshold, the service needs to be migrated to ensure the QoS. The current edge server needs to find the other edge servers that can meet the communication delay requirements, and add them to a candidate list. Afterwards, another round of service allocation should be performed based on the candidate list. 

(2) \textit{Edge server becomes over-utilized} (lines 14-22). Under this scenario, although the service delay can satisfy the delay requirement,
%the communication delay between the user and the edge server can satisfy the delay requirement, 
the edge server becomes over-utilized. Therefore, service migration is also required to optimize the resource utilization of edge servers and avoid performance degradation. 
%Services can be migrated to optimize the resource utilization of edge servers. 
To achieve this, a \textit{high migration function} $f^{h}_{m}$ is required for the migration decision. 
% The migration decision in this case requires two  functions based on probability, including $f^{l}_{m}$ (low migration function) and $f^{h}_{m}$ (high migration function).
% \begin{equation}
%     f^{l}_{m} = (1 - \frac{x}{T_{l}})^{\alpha}
%     \label{equation:lowMigration}
% \end{equation}
% $T_{l}$ is the lower threshold in resource utilization, and $\alpha$ is the shape parameter.
\begin{equation}
    f^{h}_{m} = \Big(1 + \frac{x-1}{1- T_{h}}\Big)^{\beta},
    \label{equation:highMigration}
\end{equation}
where $x$ is the resource utilization of edge server, $T_{h}$ is the upper threshold in resource utilization, and $\beta$ is the shape parameter. 

\begin{figure}[ht]
	\centering
	\includegraphics[width=0.5\textwidth]{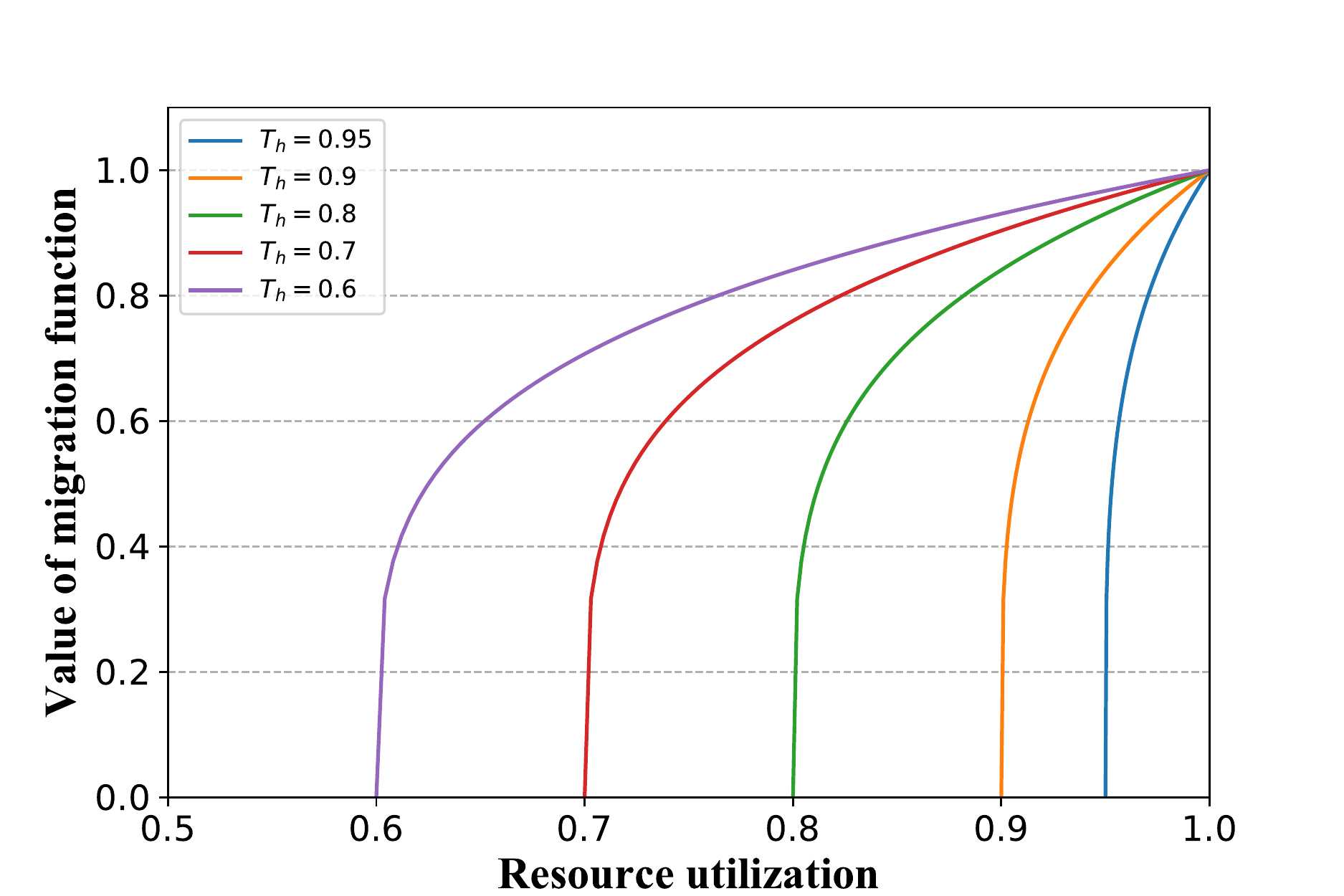}
	\caption{Migration function, $\beta=0.25$}
 	\label{fig:migrationFunction}
\end{figure}

Figure \ref{fig:migrationFunction} shows the distribution of the migration function with varied $T_h$ value. Similar to the service assignment procedure, the edge servers use the high migration function to decide whether to perform service migration. If the Bernoulli trial is successful, it means the edge server agrees to migrate the services currently running on the server to a new edge server. After the service migration decision is made, it is also necessary to select the services to be migrated, and then the algorithm performs a new round of service assignment for the selected services.

For the over-utilized edge servers, the service migration algorithm sorts all services in descending order based on resource utilization. Then we sequentially deallocate services from the server until the resource utilization of the edge server is lower than the upper threshold. For the services in the $ToMigrate$ list, the algorithm performs the service assignment procedure to assign them to new edge servers (lines 23-24).

\color{black}
\textbf{Complexity analysis of Algorithm 2:} the service migration triggered by delay violation (lines 3-9) takes a time of $O(R_m\cdot J)$, where $R_m$ is the maximum number of edge services allocated on edge servers, and $J$ is the number of edge servers; allocating the migrated services (lines 10-12, lines 23-24) takes a time of $O(R\cdot J\log(J))$, where $R$ is the maximum number of services in the whole system; the service migration triggered by over-utilized edge servers (lines 14-22) takes a time of $O(J\cdot R_m^2\log(R_m))$. Therefore, the time complexity of Algorithm 2 is $ O(R_m\cdot J+ 2R\cdot J\log(J)+ J\cdot R_m^2\log(R_m))$. 
	\color{black}

\color{black}
To be noted, our approach supports the smooth connection of services by applying service replication, which is quite similar to the VM migration process in cloud computing. When an edge service is going to be migrated, a copy of edge service will be first replicated to another edge while the original edge service is still running and connecting with mobile user. When the replicated edge service is ready, the connection will be switched from the original one to the migrated one. After the mobile user connects to the migrated edge service, the original edge service can be destroyed if no user connects to it.
\color{black}
%\color{red}Algorithm Description. \color{black}. 

% The following are the service selection methods corresponding to the two migration types:
% \begin{itemize}
%     \item Migration with low utilization: All services running on the servers are added to the list of candidate services for migration since the idle edge servers should be shut down to reduce energy consumption, 
%      \item Migration with high utilization: Sort all services in descending order based on resource utilization and sequentially add them to the candidate list until the resource utilization of the edge server is lower than the upper threshold.
% \end{itemize}

% the migration costs need to be considered when assigning a new edge server.
% the original edge server to the new server. 

\begin{algorithm2e}[h]
    \caption{Service Migration Algorithm}
    \KwIn{Mobile User $U$, Edge Service $S$, Edge Server List $ES\_List$, Delay Threshold $T_d$}

    $ToMigrate = []$  \\\vspace{10pt}

    // (1) Delay violates the threshold.\\\vspace{5pt}
    % (1) 将不能满足时延要求的服务进行迁移

    \For{$ES_i\ in\ ES\_List$}{
        \For{$S_j\ in\ ES_i.Service\_List$}{
            $U \leftarrow S_j.user$
            
            $delay \leftarrow getDelay(U, S_j, ES_i)$ according to Eq. (2)
            
            \If{$delay >= T_d$}{
                $ToMigrate \leftarrow ToMigrate \cup S_j$
                
                Deallocate $S_j$ from $ES_i$
            }
        }
    }

    \For{$S_i\ in\ ToMigrate$}{
        $ServiceAssignment(S_i)$ by using Algorithm 1
    }
    $ToMigrate.clear()$ \\\vspace{10pt}

    // (2) Migrate services from over-utilized edge servers.\\\vspace{5pt}
    % $ExcludeESList = []$

    \For{$ES_i\ in\ ES\_List$}{
        \If{$overUtilized(ES_i) == true$}{
            $result \leftarrow getMigrationDecision(ES_i)$ based on Eq. (9)
            
            \If{$result == Accept$}{
                
                % $ExcludeESList.add(ES_i)$
                % $ES_i.serviceList().sortUtilization()$
                Sort $ES_i.Service\_List$ by CPU Utilization
                
                \While{$overUtilized(ES_i) == true $ }{
                    $S \leftarrow getService(ES_i)$
                    
                    $ToMigrate \leftarrow ToMigrate \cup S$
                    
                    Deallocate $S$ from $ES_i$
                }
            }
        }
    }
    % \ForEach{Service$s_i\ in\ ServiceCandidate$}{
    %     $Assign(s_i, ExcludeESList)$
    % }
    \For{$S_i\ in\ ToMigrate$}{
        $ServiceAssignment(S_i)$ by using Algorithm 1
    }
\end{algorithm2e}

\subsection{PDMA Competitive Analysis} 
We apply competitive analysis to analyze our proposed approach based on probabilistic management for services on edge servers. 
We assume that there are $J$ heterogeneous edge servers, and $R$ heterogeneous services. The communication time between the user and edge server from the original connection and new connection (after migration) is denoted as $t_c$ and $t_c'$. The corresponding connection cost per unit time are denoted as $C_e$ and $C_e'$. The processing time of the original edge server is $t_p$, and the processing time of the migrated edge server is $t_p'$. The processing cost per unit time for the original edge server and migrated edge server are denoted as  $C_p$ and $C_p'$. Let $t_m$ be the migration time and $C_m$ be the migration cost per unit time. Without loss of generality, we can define $t_cC_e = 1$, $t_pC_p = \varepsilon$ and $t_mC_m = \delta$.  Let $\tau$ be the times of migration that happens during the observation time. 
\begin{theorem}	
	The upper bound of the competitive ratio of PDMA algorithm for the edge service migration is $\frac{PDMA(U)}{OPT(U)} \leq 1 + \frac{(2 + \varepsilon + \delta) JR}{(1+ \varepsilon + \delta)(J+R)}$.
\end{theorem}

\begin{proof}
	Under the normal status, the number of services deployed on edge servers is ${R}/{J}$, while in QoS violated or overloaded situation, at least ${R}/{J}+1$ services are deployed to a single edge server. Thus, the maximum number of QoS violated nodes is $J_o = \lfloor\frac{R}{{R}/{J}+1}\rfloor$, which is  equivalent to $J_o = \lfloor{R/J+R}\rfloor$.

For a set of users $U$, the optimal offline algorithm for problem only keeps the services on edge servers and migrates minimum services, thus the total cost of an optimal offline algorithm is defined as:
\begin{equation}
	OPT(U) = \tau(t_cC_eJ+t_pC_pJ+t_mC_mJ).
\end{equation}

For our proposed approach, the total cost with migration can be defined as below:
\begin{equation}
%\footnotesize
PDMA(U) = \tau\{t_cC_e(J+J_o) + t_c'C_e'J_o + t_pC_pJ + t_p'C_p'J_o + t_mC_m(J+J_o)\}.
\end{equation} 

According to our proposed approach, the communication cost that user connect with the migrated node should be no more than the orignal node, thus $ t_c'C_e' \leq t_cC_e\ $. And the processing cost of migrated node is no more than the orignal node, thus $t_p'C_p' \leq t_pC_p$. Then we have:
\begin{equation}
PDMA(U) \leq \tau\{t_cCe(J+2J_o) + t_pC_p(J+J_o) + t_mC_m(J+J_o)\}.
\end{equation}
Therefore, the competitive ratio of an optimal deterministic algorithm as: 
\begin{equation}
%\small
\begin{split}
\frac{PDMA(U)}{OPT(U)} & \leq \frac{\tau\{t_cCe(J+2J_o) + t_pC_p(J+J_o) + t_mC_m(J+J_o)\}}{\tau(t_cC_eJ+t_pC_pJ+t_mC_mJ)} \\
& = \frac{t_cC_eJ+t_pC_pJ+t_mC_mJ + (2t_cC_e+t_pC_p+t_mC_m)J_o}{t_cC_eJ+t_pC_pJ+t_mC_mJ} \\
& = 1 + \frac{(t_cC_e+t_pC_p+t_mC_m)J_o + t_cC_eJ_o}{(t_cC_e+t_pC_p+t_mC_m)J} \\
& = 1 + \frac{J_o}{J} +\frac{t_cC_eJ_o}{(t_cC_e+t_pC_p+t_mC_m)J}.
\end{split}
\end{equation}
 As $J_o = \lfloor\frac{JR}{J+R}\rfloor$, we have $J_o \leq \frac{JR}{J+R}$
The competitive ratio is defined as: 
\begin{equation}
\begin{split}
\frac{PDMA(U)}{OPT(U)} & \leq 1+ \frac{J_o}{J} + \frac{J_o}{(1+ \varepsilon + \delta) J} \\
& = 1 + \frac{(2 + \varepsilon + \delta) J_o}{(1+ \varepsilon + \delta)J} \\
& \leq 1 + \frac{(2 + \varepsilon + \delta) JR}{(1+ \varepsilon + \delta)(J+R)}.
\end{split}
\end{equation}
\end{proof}

\begin{figure}[ht]
	\centering
	\includegraphics[width=0.45\textwidth]{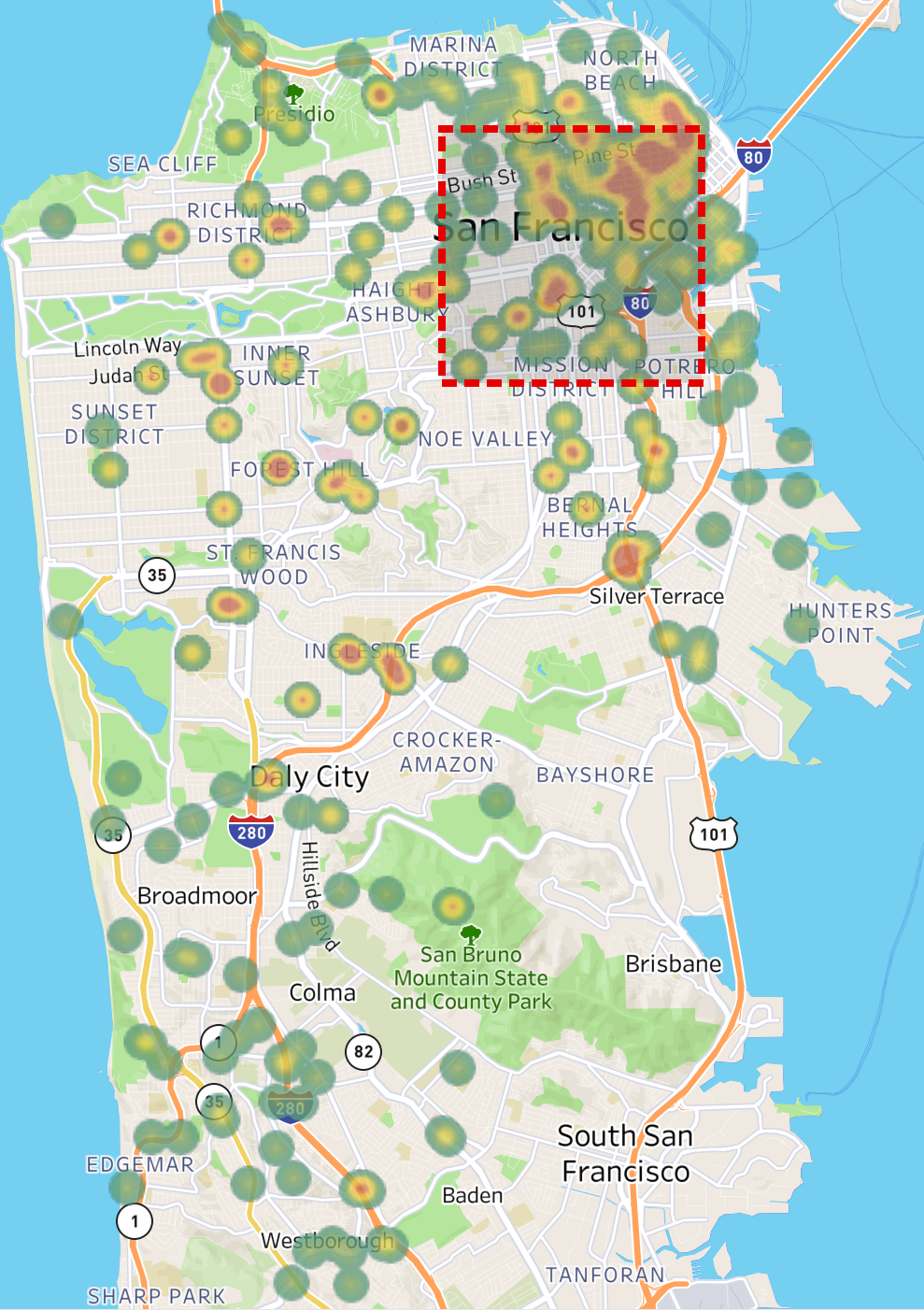}
	\caption{Base station dataset}
 	\label{fig:data_baseStation}
\end{figure}

\begin{figure}[ht]
	\centering
	\includegraphics[width=0.45\textwidth]{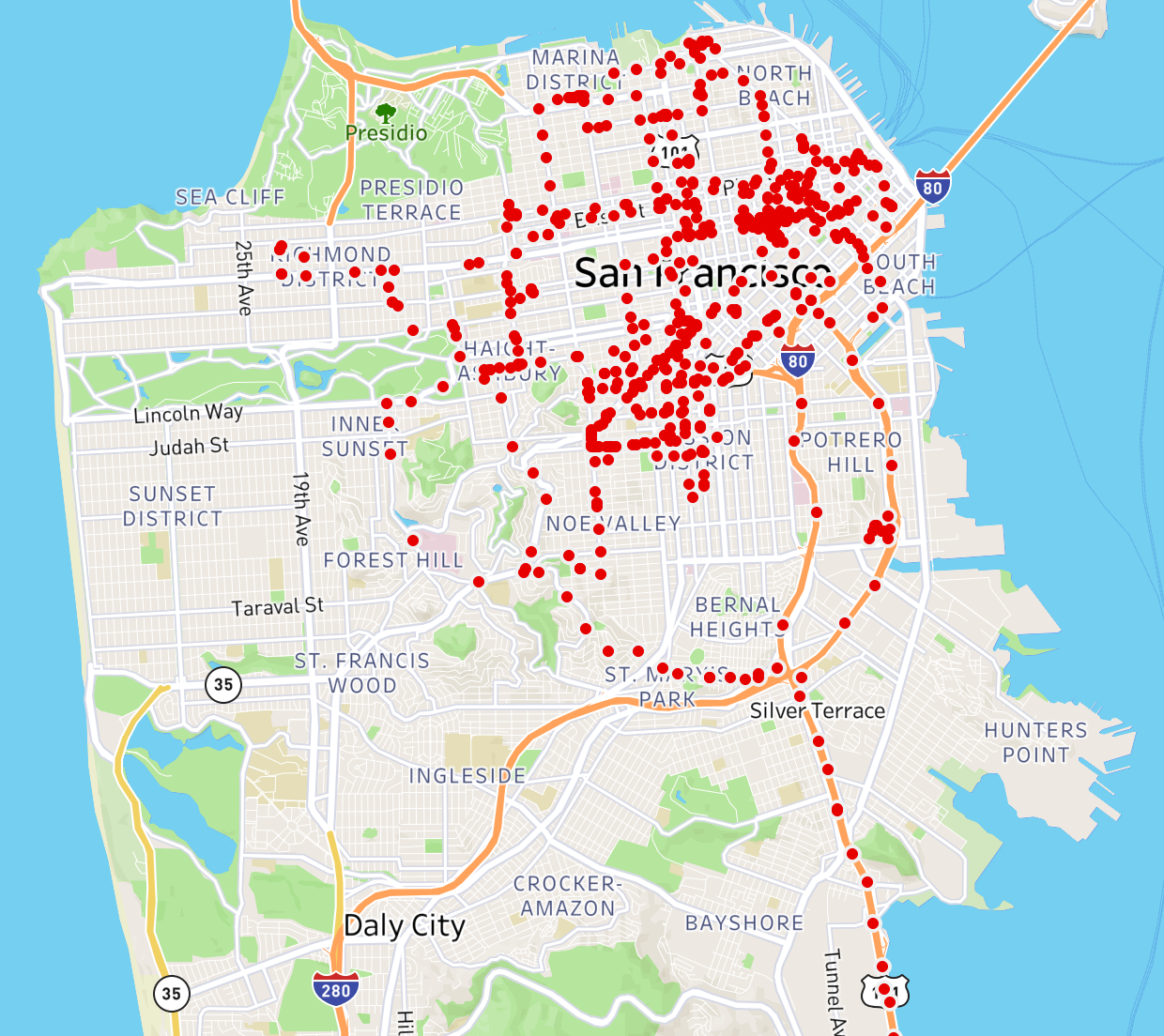}
	\caption{Taxi traces}
 	\label{fig:data_taxi&BS}
\end{figure}

%\section{Metrics}

\section{Performance Evaluations}

To evaluate algorithm performance, we simulate the service migration scenario in MEC based on iFogSim \cite{gupta2017ifogsim} and conduct experiments with several baselines. %To validate the performance, we also made a comparison with several baselines. 
To carry out the experiments, three datasets are utilized for our experiments, including (1) the location of the base stations, (2) the mobility traces of users, and (3) the workload data of edge services. We will first introduce these three datasets in this section, and then explain the configurations and procedures of our experiments. At the end of the section, we will present the evaluation results of our algorithms.

\subsection{Datasets Description}
We use three real-world datasets mentioned above to carry out our experiments.

First, we get the base station dataset from antenna distribution dataset \cite{antennasearch}, which consists of the location information of 422 base stations in San Francisco. Figure \ref{fig:data_baseStation} shows the distribution of the base station dataset. It can be observed that the density of edge servers varies in different areas, e.g., more base stations are deployed in the central business district. 
%of which we selected 200 base stations via the K-means clustering to deploy edge clouds in our experiments. Figure\ref{fig:data_baseStation} shows the base station dataset.

To simulate a real-world MEC scenario, we obtained realistic mobility traces of 536 taxis in San Francisco \cite{crawdad}. The dataset records the locations of each taxi (represented by the latitude and longitude) every 60 seconds on May 31, 2008. Each taxi in this dataset acts as a mobile user in our simulation and runs one mobile service that needs communications with edge servers while traveling around the city. The distance between base stations and mobile users can be calculated by Euclidean distance. This combines the first two datasets and also helps us to simulate an Internet of Vehicles (IoV) \cite{Ding2019} scenario. Figure \ref{fig:data_taxi&BS} depicts the trace of one taxi in the whole day. We can notice that the location of the taxi can be changed significantly during the day, which also demonstrates the need for service migration to support the user in a delay-aware and mobility-aware manner. 

%\subsubsection{PlanetLab traces}
We also simulate the workload of each service running on the edge servers based on the dataset derived from PlanetLab workload~\cite{park2006comon} that includes CPU utilization data of thousands of VMs allocated to servers. We utilize the CPU utilization of VMs to represent the CPU utilization of edge services. The utilization of edge servers will be influenced by the utilization of edge services deployed on them.   

\subsection{Rush Hour Simulation}
% 如今，人们不断涌入大城市。城市内人口十分集中，尤其是上下班高峰期，上班族在特定时间段内涌入城市内某些道路、区域，导致十分严重的交通拥堵。此时，该区域内的边缘服务器可能出现负载过高的情况。如果不对边缘服务进行良好的调度，服务的延迟将会大大增加，进而影响用户的服务体验。所以，针对高峰期的服务调度是十分关键的。

% 我们在实验中对数据集在空间、时间上进行限制，希望模拟出城市上下班高峰期的场景，并在该场景中验证我们的算法的有效性。
Nowadays, the population in the large city is quite concentrated, especially during the rush hour, e.g., 8:00 am to 9:00 am on the weekday's morning. When a large volume of mobile users rush into certain roads and areas in the city, resulting in serious traffic congestion. During the rush hour, the edge servers in the crowded areas are more prone to be overloaded compared with non-rush hours. If the edge services are not properly scheduled, the delay of the services will be greatly increased, which will affect the quality of experience of users. 
Therefore, attention should be paid to the edge service management in rush hour.

To simulate the service scheduling in rush hours, based on the original workloads, we select a period of time as rush hour and a crowded area to simulate the scenario of the city during the rush hour, and evaluate the effectiveness of our scheduling algorithm.
We first utilize the K-means clustering algorithm \cite{WANG2019JPDC} on the mobility traces of taxis to select a location with the highest density of taxis. Then, we use this location as the center to frame a square ($4km \times 4km$) as the congested area in the rush hour. We then extract the data of 147 base stations in this area from the whole dataset. The red square in Figure~\ref{fig:data_baseStation} shows the selected area, which represents a much more dense distribution of mobile users than other areas. 
Afterwards, to choose the rush hour, we count the number of taxis in this area in different time periods and pick three hours of May 31, 2008 with the maximum number of taxis. After that, we extract the mobility traces of taxis moving in this congested area during the rush hour. 

For this part of the experiment, we generate new workload traces for edge services derived from the original PlanetLab dataset. %The CPU utilization of edge services within each scheduling interval is randomly generated in the range 70\% to 100\%, 
As the original resource utilization is low in PlanetLab dataset, we consider multiple edge services are connected so that they should be deployed together, in order to increase the resource utilization of edge servers and thus simulating the heavy workloads during the rush hour. This assumption conforms to the motivation of microservice architecture~\cite{Xu2020TSUSC} that can be applied to the MEC environment. 
Based on the above steps, we can perform simulations for the rush hour scenario, and the results will be demonstrated in the following sections.

\subsection{Experiment Configurations}
We conducted all our experiments on the same computer with iFogSim. The experimental configurations are as below:

For Eq.~(\ref{eq:transRate}), we set the channel bandwidth $W$ to be 20 Mhz and transmitted power of taxi $S_p$ to be 0.5W, and the noise power $N_p$ to be $2 \times 10^{-13}$ W. Besides, the wireless channel gain $g$ is set as $127+30 \log{d}$. We generate the delay matrix at every scheduling interval, and the delay of each link $m_{i,j} \in M$ is randomly generated between 5ms and 50ms.

For the configurations of edge servers, each server has 8 CPU cores with Millions of Instructions Per Second (MIPS) of 2000, 3000, and 4000, 80GB of RAM, and 10TB of storage. Each edge service is randomly configured to request 1000, 1500, 2000, 2500 of MIPS, respectively, and 8 GB of RAM.
The instructions of the task executed by each edge service are configured as 60 million.

For the algorithms, the scheduling interval is configured to 60 seconds. The delay threshold of PDMA is configured as 75ms. To evaluate the performance of our service migration algorithm, we focus on two metrics to evaluate the performance, including the \textit{Migration Cost} and the \textit{Overall Delay} based on the optimization objectives. In addition, we also record the \textit{Number of Overloaded Servers} to evaluate the overloaded situation, especially for the rush hours. The descriptions of the metrics are as below: 

\begin{itemize}
\item \textbf{Overall delay} represents the average delay of all services during the experiments as represented in Eq. (\ref{eq:overallDelay}). 

\item \textbf{Migration cost} is the sum of the cost of all service migrations as represented in Eq. (\ref{eq:migrationCost}). In our simulation, the function $F$ calculates the distance between the source server and destination server.

\item \textbf{Average number of overloaded edge servers} is applied to evaluate the effect of algorithms to relieve the overloaded situation. The overloaded hosts are identified based on the predefined utilization threshold. {In our experiments, we set the utilization threshold as 0.9, as this value has been used widely to identify overloaded hosts in data centers.}
\end{itemize}

We also compare our approach with three scheduling algorithms.
\begin{itemize}
    \item{\textbf{Nearest edge server first (NF)}}: it assigns the edge service to the edge server closest to the mobile user based on distance, which can reduce communication costs.
    
    \item{\textbf{Never migrate (NM)}}: it never migrates the edge services, thus the migration costs can be reduced.

     \item{\textbf{Top-K}}: it sorts all the edge servers by their CPU utilization and selects one randomly from the top K busiest servers to host a new coming or migrated edge service randomly. Here, $K$ is configured as $0.1\times J$, where $J$ is the size of edge servers.
     
     {\color{black}\item{\textbf{CHERA}\cite{zhao2019optimal}}: it is a clustering-based heuristic algorithm for edge resource allocation. It adopts a clustering procedure to allocate applications to suitable edge servers and minimizes the average service response time.}
\end{itemize}

%We conducted several rounds of experiments to evaluate all these algorithms and made a comparison between them.
%We  also adopt three metrics to evaluate the algorithm performance, including $Overall\ Delay$, $Migration\ Cost$  and $Average\ Number\ of\ Overloaded\ Servers$.

\subsection{Experiments and Results}

We divide the experiments into two parts, including the experiments on PlanetLab traces and the experiments on Rush Hour traces. In each part, we investigate two parameters on algorithms performance. The first one is $distance\ threshold$ to represent the coverage of base stations. If the distance between a mobile user and a base station exceeds the distance threshold, the user will not be able to access the base station and thus the service connection should be switched to another base station, which will affect the service delay. 
The second one is \textit{ratio of clients and servers number}  to demonstrate the scalability of scheduling algorithms. Since the PlanetLab traces specify the number of clients, we reduce the number of edge servers to modify the ratio. And, for the Rush Hour simulations, we increase the volume of mobile users in the crowded area to simulate the scenario.

\subsubsection{Results with PlanetLab Traces}
We first present the experiment results based on PlanetLab traces by varying the coverage of the base station and number of edge servers.

\noindent \textbf{(a) Varied coverage of base stations}

We configure the number of mobile users to be 1000 and vary the distance threshold from 200m to 2000m. As shown in Figure~\ref{fig:PlanetLab_PerformanceVaryingThreshold}, the distance threshold has a slight impact on the three metrics. Figure~\ref{fig:PlanetLab_distanceThreshold_delay} shows the results of overall delay.
\color{black} {\color{black}CHERA achieves close delay with} PDMA and performs lower delay than NF and others. The overall delay of PDMA is lower than 62ms when the coverage of the base station is longer than 200m. {\color{black}The migration cost of PDMA is much lower than that of CHERA}, indicating that PDMA can decrease the communication overhead during the service migration. Since the CPU utilization of PlanetLab and the ratio of mobile clients to edge servers is low, very few edge servers become overloaded during the simulation.
\color{black}
\noindent \textbf{(b) Varied number of edge servers}

The average volume of mobile users in PlanetLab traces is 1000. To modify the ratio of clients and servers number, we configure the number of edge servers to be from 100 to 400. 
As shown in Figure~\ref{fig:PlanetLab_PerformanceVaryingHostNum}, we can notice that PDMA is able to maintain good performance when there are fewer servers. 
\color{black}{\color{black} When there are 100 edge servers, the overall delay of PDMA is close to that of CHERA and is 13.5ms lower than that of NM.}
{\color{black} The migration cost of PDMA is only 4.3\%, 15.8\% and 30.3\% compared with Top-K, CHERA and NF, respectively. For the number of overloaded edge servers, less than 5 servers are over-utilized when using PDMA, while CHERA can lead to more than 12.}

% We fix the distance threshold to 1000 and change the number of mobile users from 400 to 2000.

\begin{figure*}[t]
	\centering
	\begin{subfigure}{0.33\linewidth}
		\centering
		\includegraphics[width=0.99\linewidth]{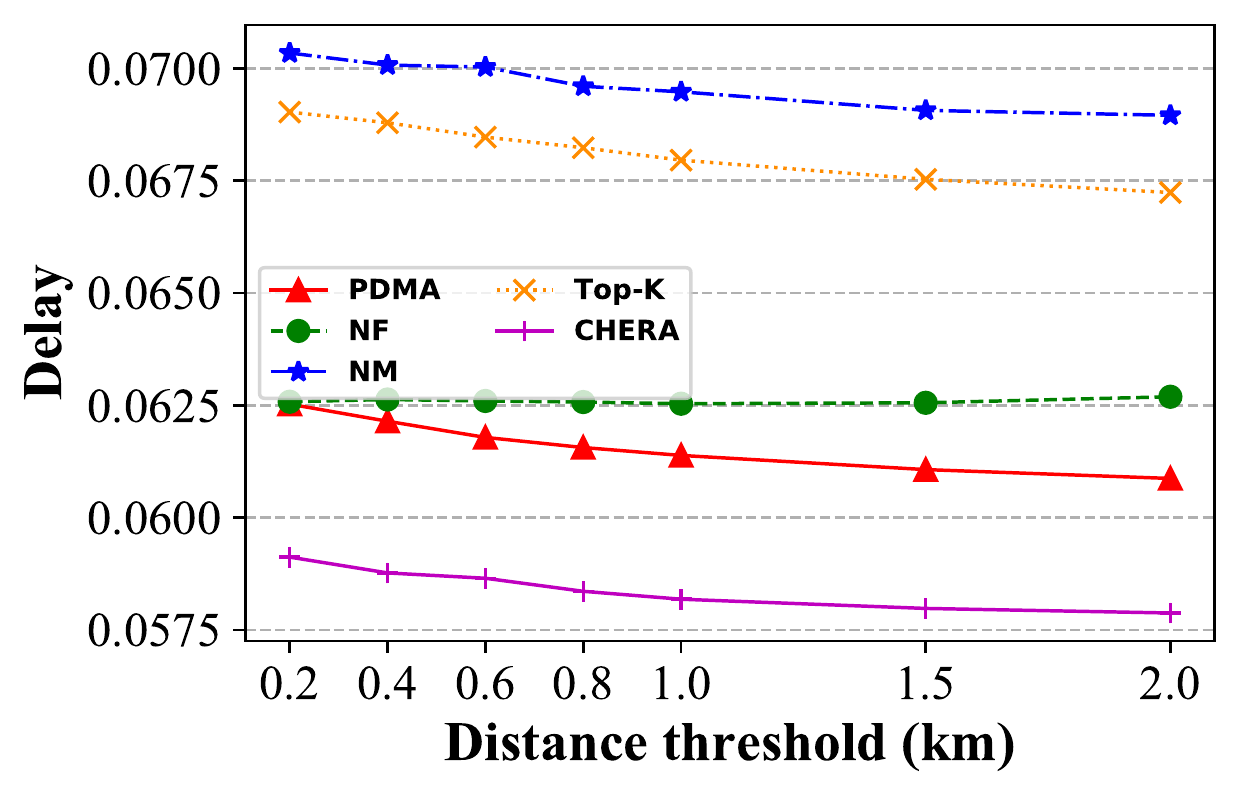}
		\captionsetup{font={large}} 
		\caption{Overall delay}
		\label{fig:PlanetLab_distanceThreshold_delay}
	\end{subfigure}
	\begin{subfigure}{0.33\linewidth}
		\centering
		\includegraphics[width=0.99\linewidth]{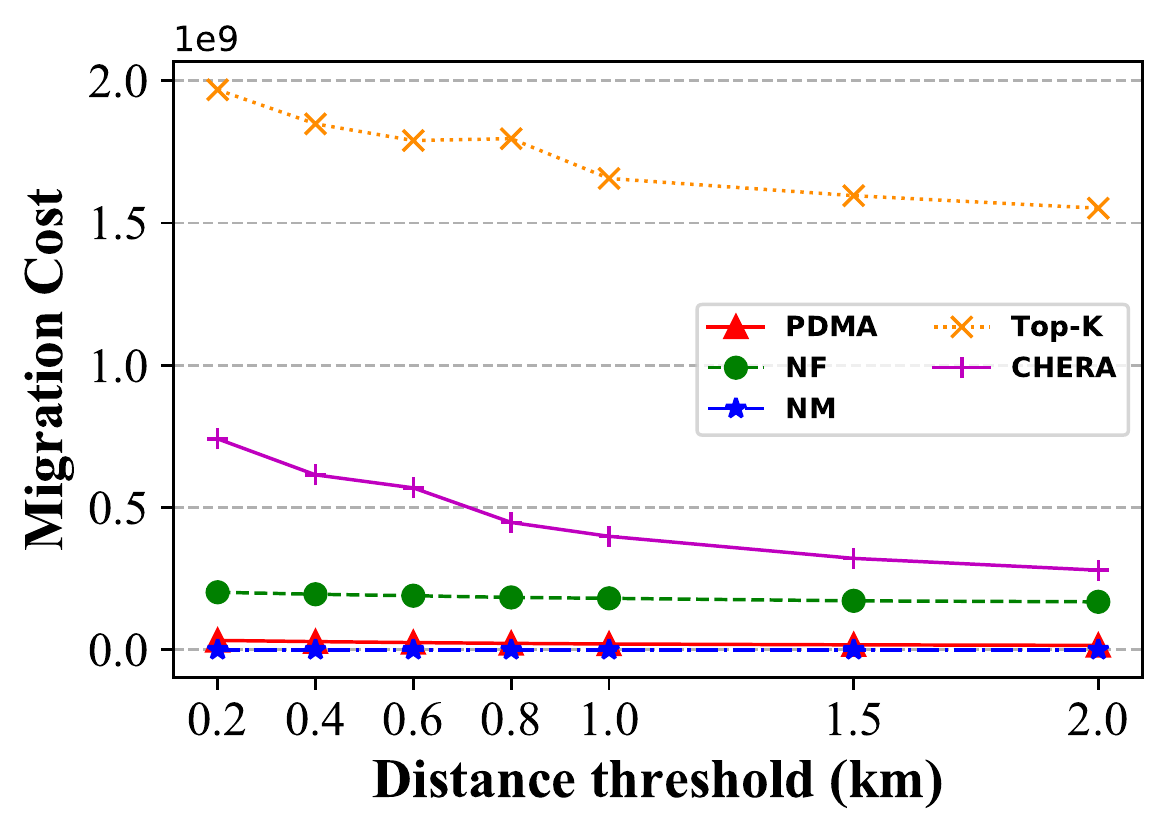}
		\captionsetup{font={large}} 
		\caption{Migration cost}
		\label{fig:PlanetLab_distanceThreshold_cost}
	\end{subfigure}
	\begin{subfigure}{0.33\linewidth}
	\centering
	\includegraphics[width=0.99\linewidth]{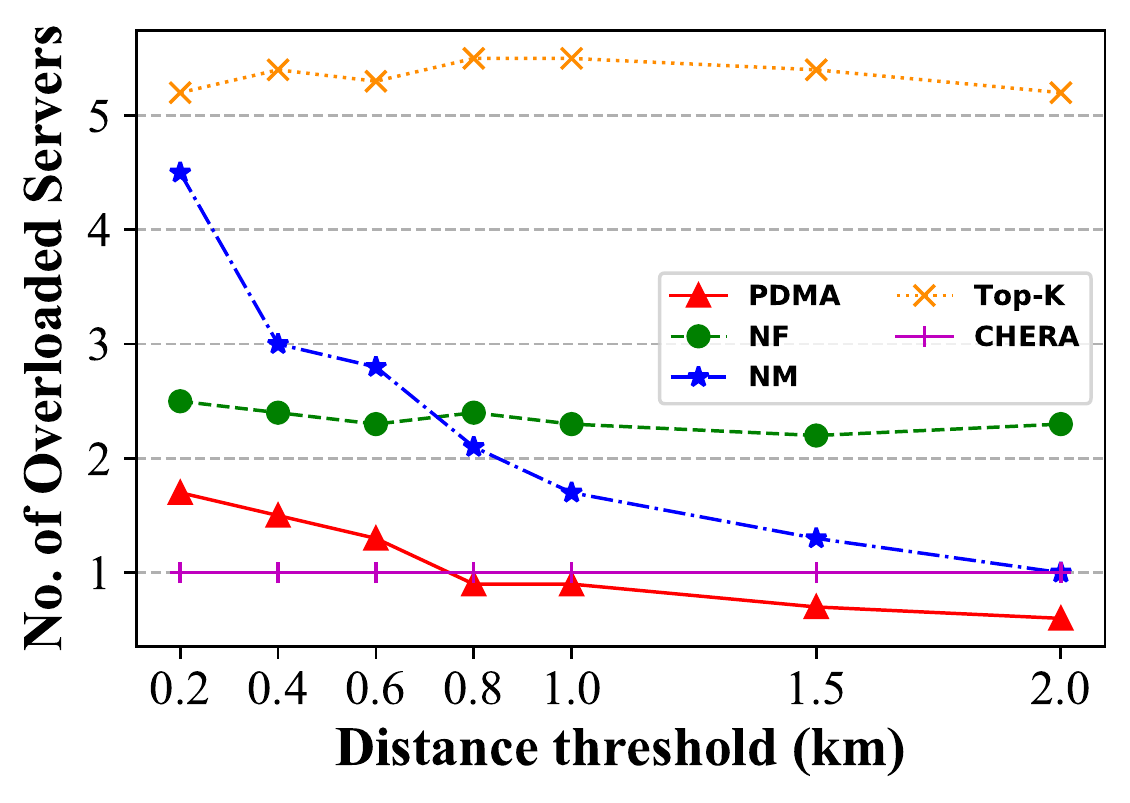}
	\captionsetup{font={large}} 
	\caption{Number of Overloaded Servers}
	\label{fig:PlanetLab_distanceThreshold_overloaded}
    \end{subfigure}
    % \captionsetup{font={large}} 
	\caption[VarPerOptCom]{PlanetLab:\ Performance comparison of algorithms with varied distance thresholds}
	\label{fig:PlanetLab_PerformanceVaryingThreshold}
\end{figure*}

\begin{figure*}[t]
	\centering
	\begin{subfigure}{0.33\linewidth}
		\centering
		\includegraphics[width=0.99\linewidth]{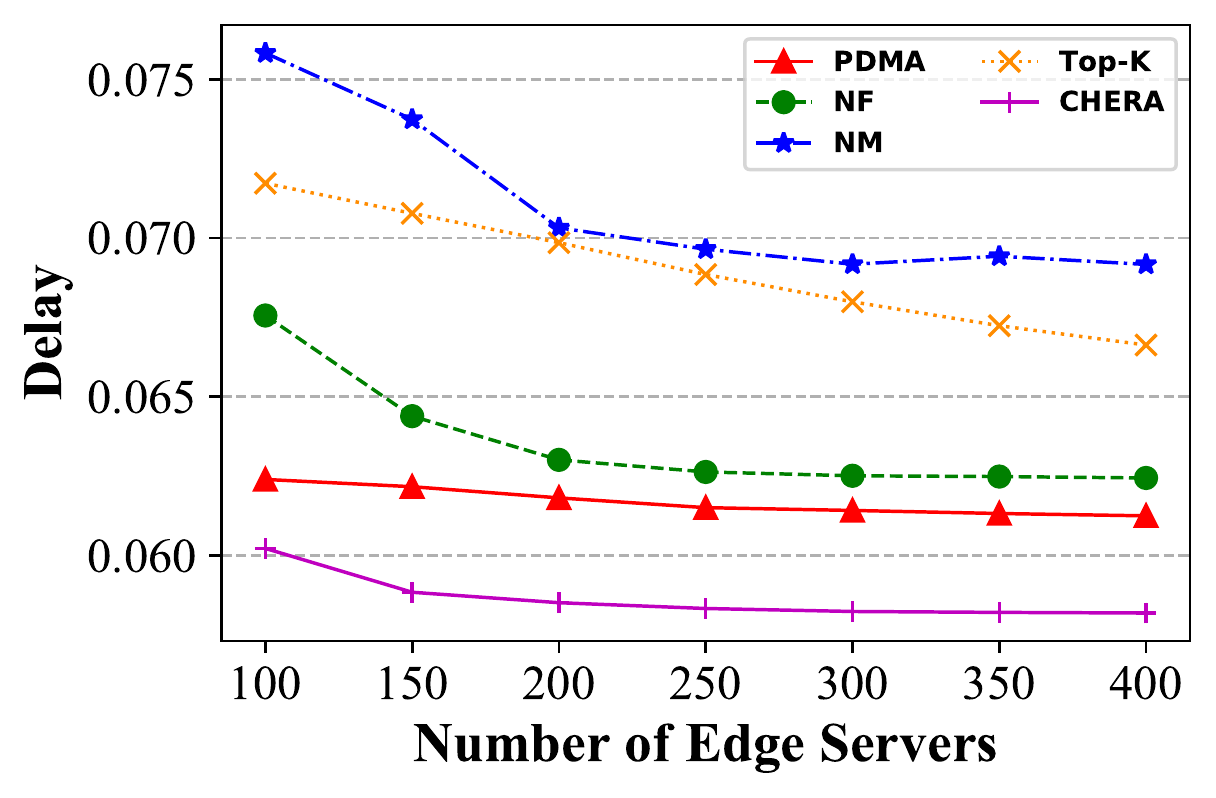}
		\captionsetup{font={large}} 
		\caption{Overall delay}
		\label{fig:PlanetLab_hostNum_delay}
	\end{subfigure}
	\begin{subfigure}{0.33\linewidth}
		\centering
		\includegraphics[width=0.99\linewidth]{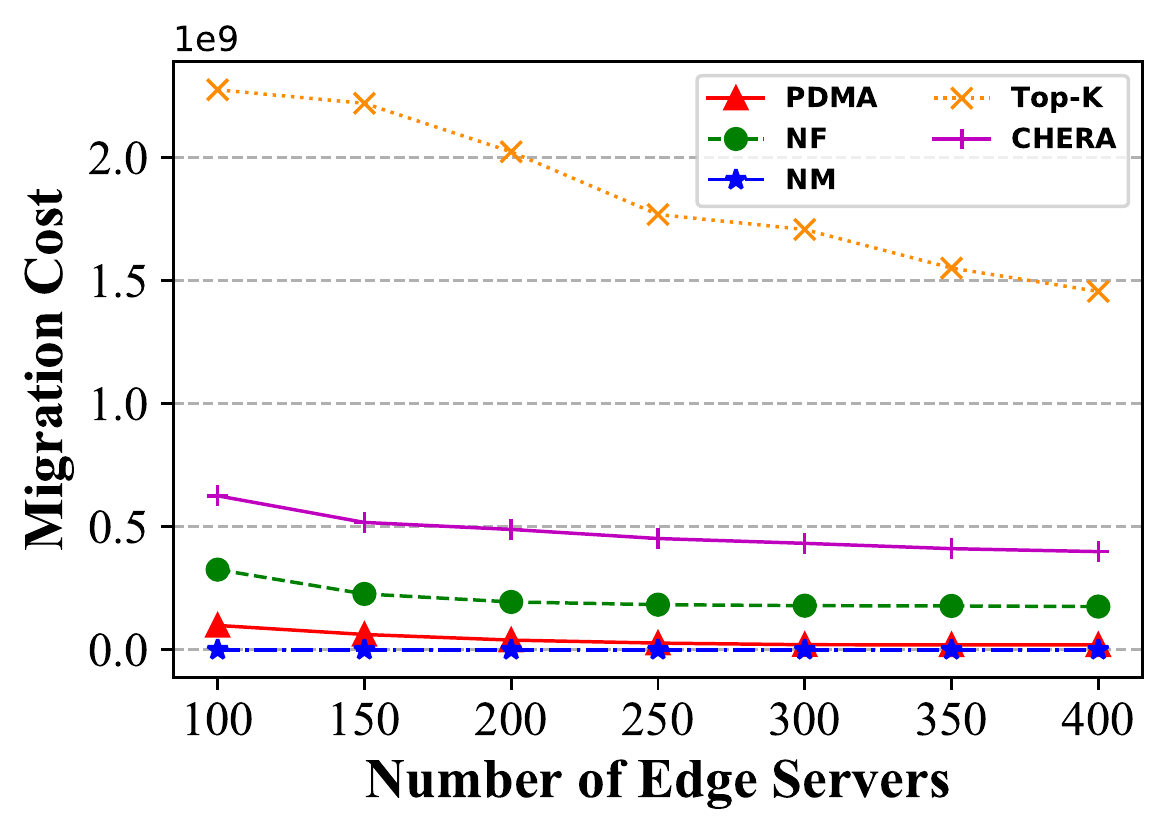}
		\captionsetup{font={large}} 
		\caption{Migration cost}
		\label{fig:PlanetLab_hostNum_cost}
	\end{subfigure}
	\begin{subfigure}{0.33\linewidth}
	\centering
	\includegraphics[width=0.99\linewidth]{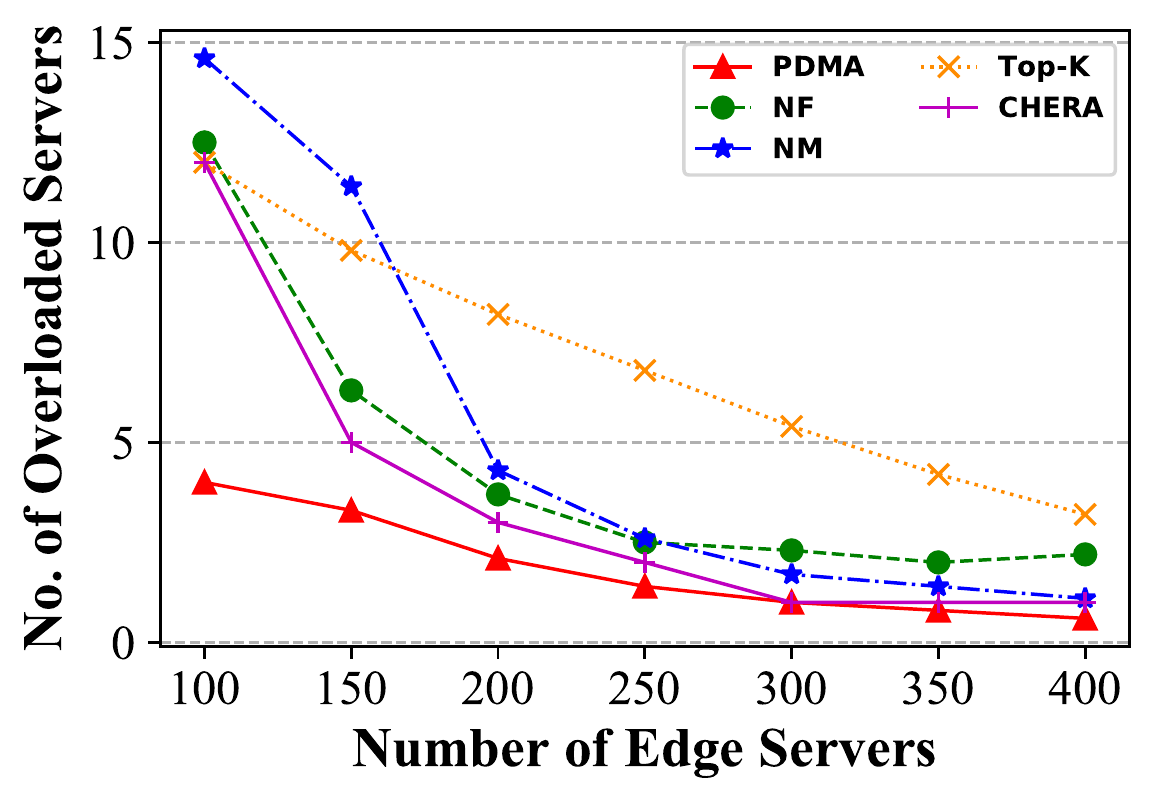}
	\captionsetup{font={large}} 
	\caption{Number of Overloaded Servers}
	\label{fig:PlanetLab_hostNum_overloaded}
    \end{subfigure}
    % \captionsetup{font={large}} 
	\caption[VarPerOptCom]{PlanetLab:\ Performance comparison of algorithms with varied number of edge servers}
	\label{fig:PlanetLab_PerformanceVaryingHostNum}
\end{figure*}

\color{black}
\subsubsection{Results with Rush Hour Traces}
We present the experiment results of rush hour traces in this part. First, we only use the data of base stations in the selected area and deploy one edge server on each base station. Then, we configure the two parameters with varied values and run several rounds of simulations.

\noindent \textbf{(a) Varied coverage of base stations}

We configure the number of mobile users to be 1000 and vary the distance threshold from 200m to 2000m. \color{black}As shown in Figure~\ref{fig:RushHour_distanceThreshold_delay}, 
%PDMA achieves the lowest overall service delay among all four algorithms when changing the distance threshold. 
{\color{black}except for NM, other algorithms can perform better on overall delay when increasing the distance threshold.} It results from the decrease in the number of service migrations, which can also be observed in Figure~\ref{fig:RushHour_distanceThreshold_cost}. The migration cost of PDMA is the closest to that of NM (the cost is zero). The NF always chooses the nearest edge server to the mobile user to migrate the service and thus leading to more migration cost. Top-K performs the worst on migration cost because the top K busiest servers may be far away from the current edge server. Figure \ref{fig:RushHour_distanceThreshold_overloaded} shows the number of overloaded servers when utilizing different algorithms. {\color{black}PDMA controls the number of overloaded servers to be less than 20 while CHERA incurs nearly 40 overloaded edge servers and other algorithms incur more than 50.} The reason is that PDMA will not select those edge servers that are likely to be overloaded and thus avoiding more overloaded situations.

\color{black}
\noindent \textbf{(b) Varied volume of mobile users}

In order to evaluate the performance of scheduling algorithms in the rush hours, we increase the number of mobile users to simulate weekday's morning when more mobile users enter the crowded area.
We fix the distance threshold as 1000m and vary the volume of clients from 200 to 1000. \color{black}As shown in Figure~\ref{fig:RushHour_clientNum_overloaded}, the higher volume of mobile clients causes more edge servers to be overloaded, leading to performance degradation and higher computation delay. Therefore, it can be noticed in Figure~\ref{fig:RushHour_clientNum_delay} that the overall delay also becomes higher when the number of clients increases. In addition, a higher service delay that exceeds the delay threshold will trigger more migrations, increasing the migration cost and the migration downtime.

When the number of clients increases to be more than 400, PDMA achieves better overall delay than NF. {\color{black}The overall delay of PDMA is 64ms, which is 2ms more than that of CHERA, and 6ms, 16ms less than that of NF and NM, respectively, when the volume of clients is 1000.} The migration cost of PDMA is also maintained at a very low level. {\color{black} For example, when there are 1000 clients, CHERA produces 4.6 times more migration cost than PDMA.} Compared with other algorithms, PDMA is able to prevent edge servers from becoming overloaded. As we can observe from Figure~\ref{fig:RushHour_clientNum_overloaded}, less than 20 edge servers become over-utilized when there are 1000 mobile clients, which {\color{black} is 55\% lower than that of CHERA and} is 75\% lower than those of NM and Top-K. {\color{black} Less over-loaded servers can reduce edge servers' energy consumption and also prevent edge servers from performance degradation during the rush hour.}
In addition, PDMA has the lowest increase rate in terms of the three adopted metrics among all algorithms, which validates the scalability of PDMA that it can be adapted to the scenarios with clients rushing into the MEC system while ensuring user experience.

\color{black}

%%%% Rush Hour %%%% 
%%%% distance threshold
\begin{figure*}[t]
	\centering
	\begin{subfigure}{0.33\linewidth}
		\centering
		\includegraphics[width=0.99\linewidth]{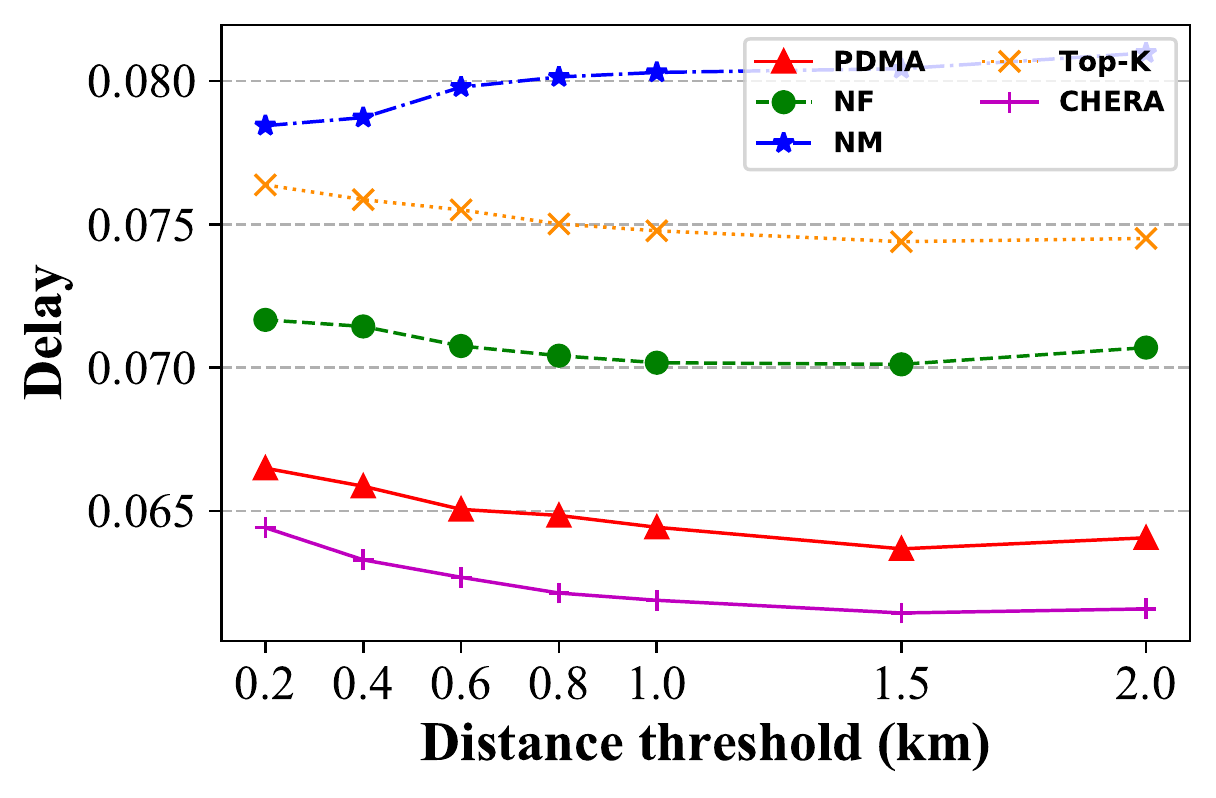}
		\captionsetup{font={large}} 
		\caption{Overall delay}
		\label{fig:RushHour_distanceThreshold_delay}
	\end{subfigure}
	\begin{subfigure}{0.33\linewidth}
		\centering
		\includegraphics[width=0.99\linewidth]{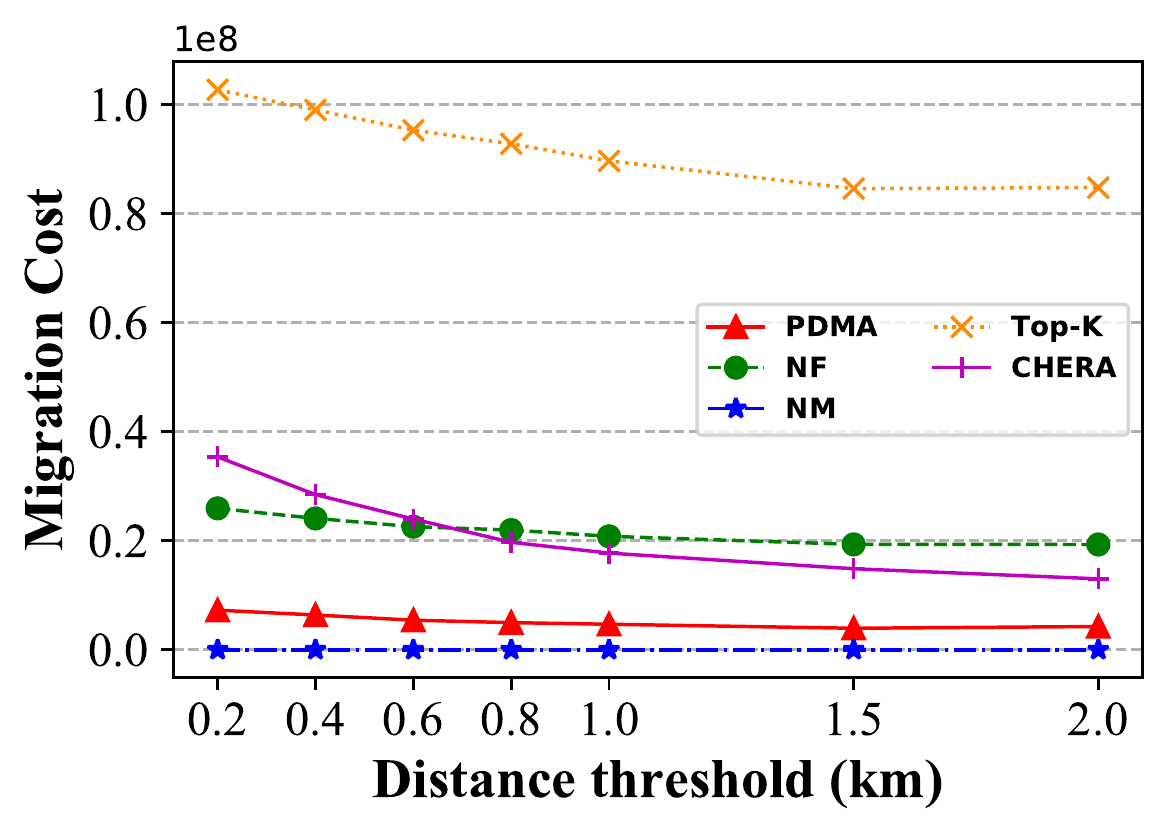}
		\captionsetup{font={large}} 
		\caption{Migration cost}
		\label{fig:RushHour_distanceThreshold_cost}
	\end{subfigure}
	\begin{subfigure}{0.33\linewidth}
	\centering
	\includegraphics[width=0.99\linewidth]{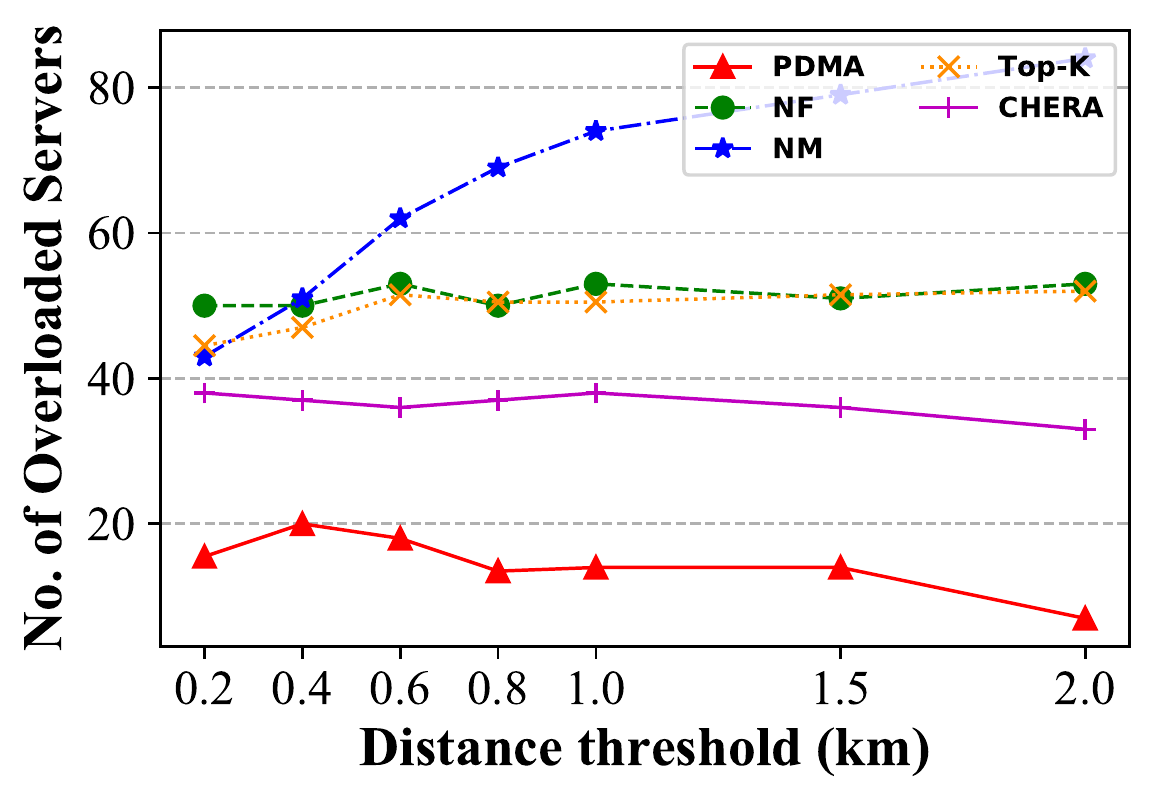}
	\captionsetup{font={large}} 
	\caption{Number of Overloaded Servers}
	\label{fig:RushHour_distanceThreshold_overloaded}
    \end{subfigure}
    % \captionsetup{font={large}} 
	\caption[VarPerOptCom]{Rush Hour:\ Performance comparison of algorithms with varied distance thresholds}
	\label{fig:RushHour_PerformanceVaryingThreshold}
\end{figure*}

%%%% client number
\begin{figure*}[t]
	\centering
	\begin{subfigure}{0.33\linewidth}
		\centering
		\includegraphics[width=0.99\linewidth]{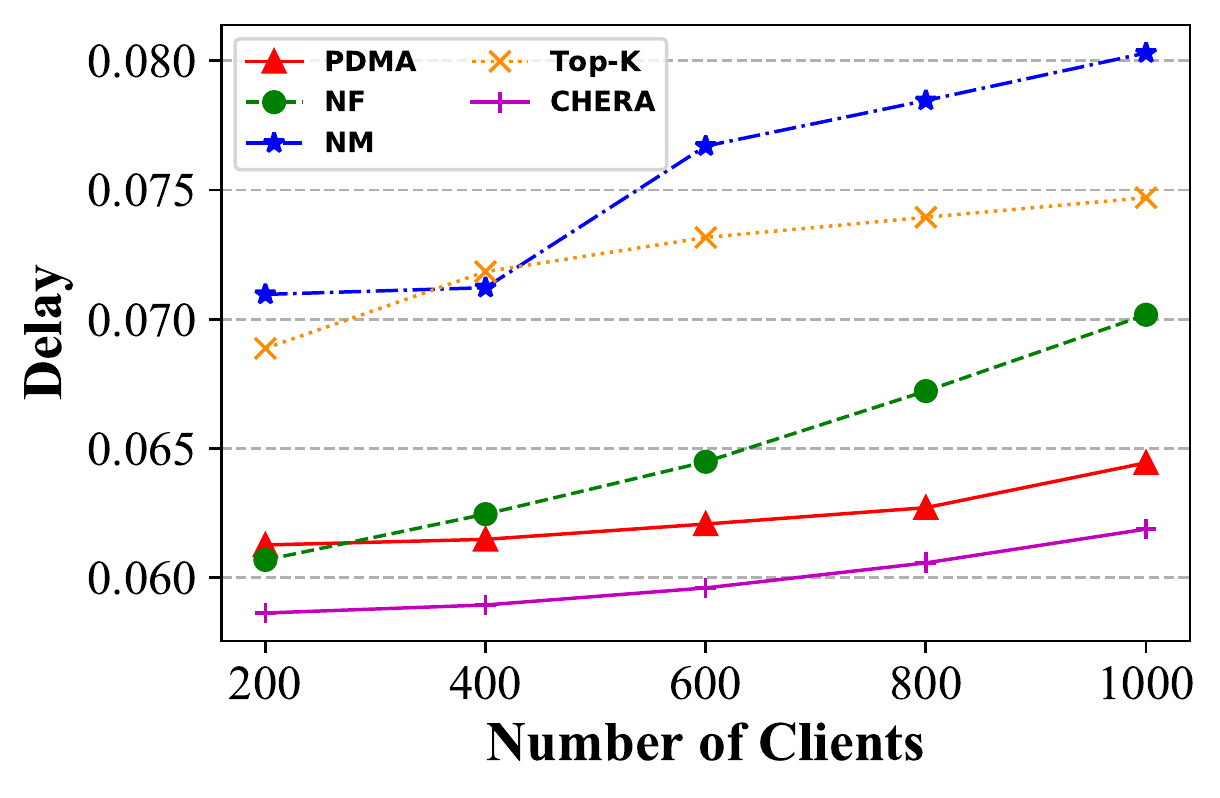}
		\captionsetup{font={large}} 
		\caption{Overall delay}
		\label{fig:RushHour_clientNum_delay}
	\end{subfigure}
	\begin{subfigure}{0.33\linewidth}
		\centering
		\includegraphics[width=0.99\linewidth]{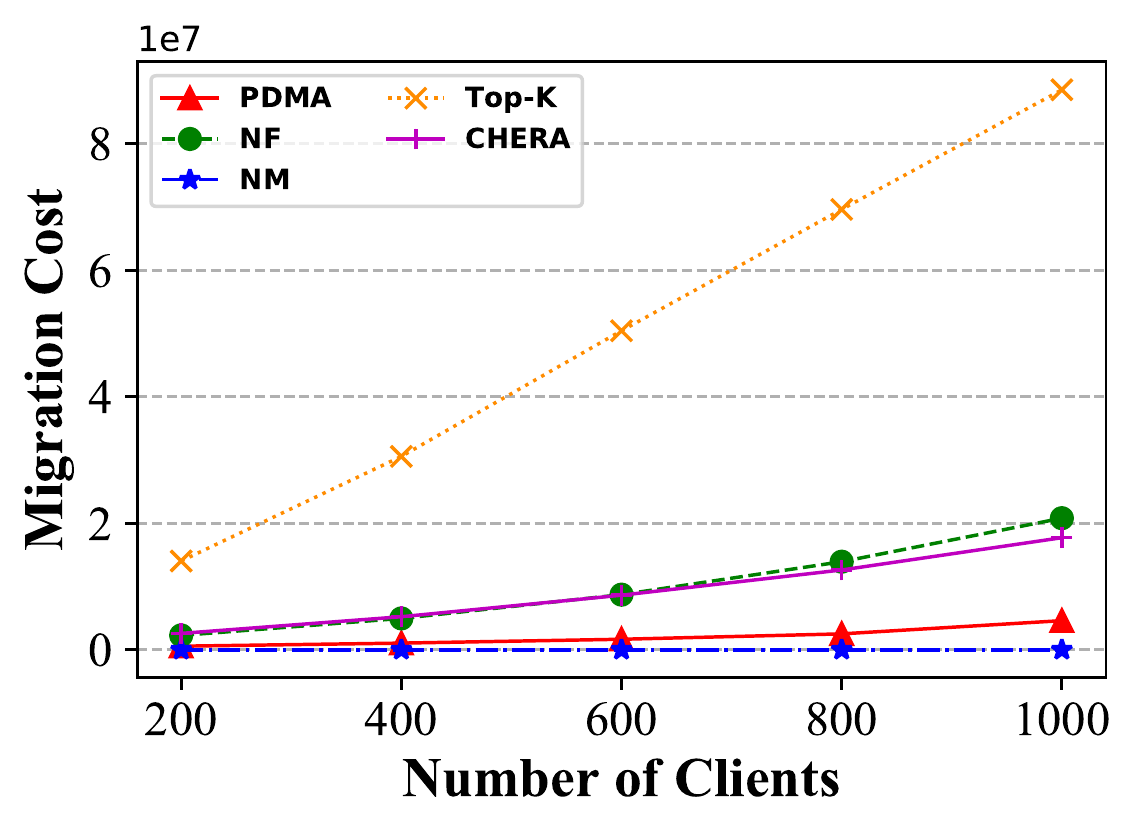}
		\captionsetup{font={large}} 
		\caption{Migration cost}
		\label{fig:RushHour_clientNum_cost}
	\end{subfigure}
	\begin{subfigure}{0.33\linewidth}
	\centering
	\includegraphics[width=0.99\linewidth]{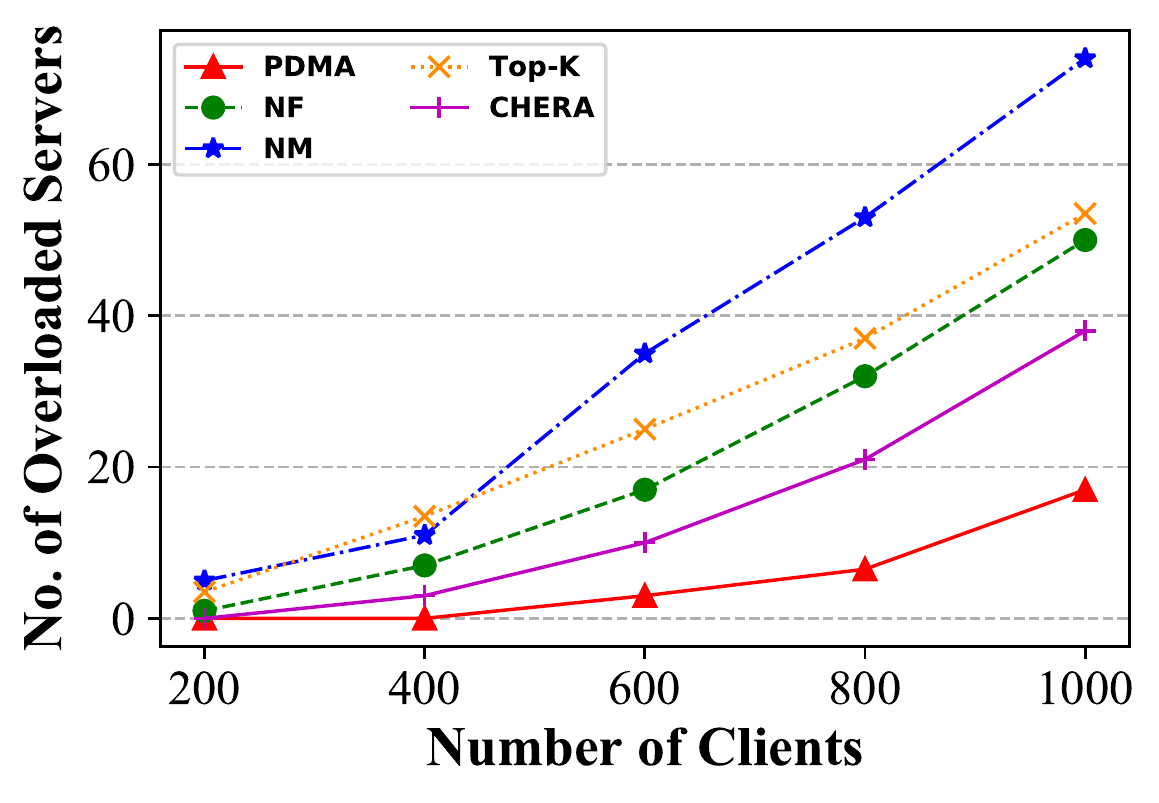}
	\captionsetup{font={large}} 
	\caption{Number of Overloaded Servers}
	\label{fig:RushHour_clientNum_overloaded}
    \end{subfigure}
	\caption[VarPerOptCom]{Rush Hour:\ Performance comparison of algorithms with varied volume of clients}
	\label{fig:RushHour_PerformanceVaryingClientNum}
\end{figure*}

% \begin{figure*}[t]
% 	\centering
% 	\begin{subfigure}{0.49\linewidth}
% 		\centering
% 		\includegraphics[width=0.99\linewidth]{./figure/varyingHostNum/Delay.pdf}
% 		\caption{Overall delay}
% 		\label{fig:DelayVaryingHostNum}
% 	\end{subfigure}
% 	\begin{subfigure}{0.49\linewidth}
% 		\centering
% 		\includegraphics[width=0.99\linewidth]{./figure/varyingHostNum/OverallCost.pdf}
% 		\caption{Migration cost}
% 		\label{fig:CostVaryingHostNum}
% 	\end{subfigure}
% 	\caption[VarPerOptCom]{Performance comparison of algorithms varying the number of edge servers}
% 	\label{fig:PerformanceVaryingHostNum}
	
% \end{figure*}

% \begin{figure*}[t]
% 	\centering
% 	\begin{subfigure}{0.49\linewidth}
% 		\centering
% 		\includegraphics[width=0.99\linewidth]{./figure/varyingClientNum/Delay.pdf}
% 		\caption{Overall delay}
% 		\label{fig:DelayVaryingClientNum}
% 	\end{subfigure}
% 	\begin{subfigure}{0.49\linewidth}
% 		\centering
% 		\includegraphics[width=0.99\linewidth]{./figure/varyingClientNum/OverallCost.pdf}
% 		\caption{Migration cost}
% 		\label{fig:CostVaryingClientNum}
% 	\end{subfigure}
% 	\caption[VarPerOptCom]{Performance comparison of algorithms varying the number of mobile users}
% 	\label{fig:PerformanceVaryingClientNum}
	
% \end{figure*}

\color{black}
\subsection{Scalability Discussions}
In this part, we discuss the scalability of our proposed approach. The problem of optimally allocating edge services to edge servers can be modelled as a bin-packing problem with varied bin sizes and prices, where bins represent the edge servers and items are the edge services to be allocated, bins sizes are the available resources of edge servers, and prices correspond to the communication costs and migration costs when allocating edge services to the edge server. As the problem is very complex and proved to be NP-hard, achieving the optimal solution to the problem can be quite time-consuming, especially for a large MEC system with a huge number of edge servers. We have compared our approach with some deterministic algorithms, e.g., NF and Top-K, which are variants of the classical Best Fit Decreasing (BFD) algorithm that allocates services to the server with the least increased costs. The BFD algorithm is a polynomial algorithm and has been proved to use no more than $11/9\cdot OPT+1$ bins \cite{Anton2012CCPE}, where $OPT$ is the minimum theoretical number of edge servers. 

The centralized and deterministic algorithms, like BFD, can function well for the MEC system with a limited number of edge servers but can be inefficient for large-scale MEC systems considering the NP-hardness of migrating multiple edge services simultaneously. Conversely, given the probabilistic nature of PDMA, it is suitable for large-scale MEC systems. We argue that it is not necessary to send allocation requests to all the edge servers in a large MEC system, as the edge server far away are not prone to be deployed with edge services considering the mobile users are with low probability to move to the distant location within the short time. With the Bernoulli trails in our proposed approach that send allocation requests to part of servers in the system, the traffic overheads can be reduced compared with the BFD-based approaches, and edge servers are added only when strictly needed. Therefore, the required number of edge servers is close to the required number of the BFD algorithm. In addition, PDMA fits well with the large MEC system with distributed edge servers, as each service allocation request can be forwarded to the edge servers in a specific area. This allows the leverage of system heterogeneity by choosing the most cost-efficient edge edges. 

To evaluate the scalability of PDMA, we performed simulations with MEC with different user scales as shown in Figure~\ref{fig:RushHour_PerformanceVaryingClientNum}. The results confirm that as the number of users increases, the migration cost will only increase slightly. The good scalability is also confirmed by the other performance metrics. For example, the number of overloaded edge servers increases slowly with the growth of the number of users. 

\color{black}

\section{Conclusions and Future Work}
This paper addresses the NP-hardness problem of delay-aware and mobility-aware service management in the MEC environment, which is sensitive to the communication costs generated in this environment. The aim is to allocate edge services to the suitable edge servers through the initial assignment and dynamic service migration to satisfy the users in terms of response time when they are moving around. With PDMA proposed in this work, the assignment and migration of edge services are based on Bernoulli trials that decide whether the edge server will accept the deployment of a specific service based on the running status. The probabilistic and low-complexity nature of our proposed approach makes it to be efficient in an environment with a large number of edge servers and rather short execution time. Especially compared with the online learning-based approaches, which can have significant computational complexity growth when the number of servers and services increases.

A theoretical proof has been provided to illustrate that our proposed approach can be bounded to the optimal solution. Simulation results based on iFogSim also demonstrate that our proposed approach can reduce the communication delays for users and transmission costs due to the service migration. For rush hours in the urban city, the proposed approach can efficiently improve the user experience. 

As for future work, we would like to 1) investigate the proposed approach into a prototype system, 2) apply learning-based approach to predict the mobility of users, and 3) integrate offloading techniques into our model to further improve algorithm performance.  

\section*{Acknowledgment}

	This work is supported by Key-Area Research and Development Program of Guangdong Province (NO. 2020B010164003), National Natural Science Foundation of China (No. 62072451, 62072187s), and SIAT Innovation Program for Excellent Young Researchers. 
%\vspace{-0.3cm}

%\bibliographystyle{wileyNJD-AMA}
\bibliography{taxiFogWylie.bib}

\begin{thebibliography}{10}
\providecommand \doibase [0]{http://dx.doi.org/}%

\bibitem{XuCSUR2019}
Xu M, Buyya R. Brownout Approach for Adaptive Management of Resources and
  Applications in Cloud Computing Systems: A Taxonomy and Future Directions.
  {\it ACM Comput. Surv.} 2019\string; 52(1).
\newblock \href {\doibase 10.1145/3234151} {doi: 10.1145/3234151}

\bibitem{IEEEhowto:Galloway}
Galloway JM, Smith KL, Vrbsky SS. Power aware load balancing for cloud
  computing. In: Proceedings of the World Congress on Engineering and Computer
  Science. ; 2011\string: 19--21.

\bibitem{XU2019JSS}
Xu M, Buyya R. BrownoutCon: A software system based on brownout and containers
  for energy-efficient cloud computing. {\it Journal of Systems and Software}
  2019\string; 155\string: 91 - 103.
\newblock \href {\doibase https://doi.org/10.1016/j.jss.2019.05.031} {doi:
  https://doi.org/10.1016/j.jss.2019.05.031}

\bibitem{Wu2020IoTJ}
{Wu} H, {Zhang} Z, {Guan} C, {Wolter} K, {Xu} M. Collaborate Edge and Cloud
  Computing With Distributed Deep Learning for Smart City Internet of Things.
  {\it IEEE Internet of Things Journal} 2020\string; 7(9)\string: 8099-8110.

\bibitem{Antonio2019SPE}
Brogi A, Forti S, Guerrero C, Lera I. How to place your apps in the fog: State
  of the art and open challenges. {\it Software: Practice and Experience}
  2020\string; 50(5)\string: 719-740.
\newblock \href {\doibase 10.1002/spe.2766} {doi: 10.1002/spe.2766}

\bibitem{Ali2020SPE}
Shahidinejad A, Ghobaei-Arani M. Joint computation offloading and resource
  provisioning for edge-cloud computing environment: A machine learning-based
  approach. {\it Software: Practice and Experience} 2020.
\newblock \href {\doibase 10.1002/spe.2888} {doi: 10.1002/spe.2888}

\bibitem{Guo2019SPE}
Guo Y, Wang S, Zhou A, Xu J, Yuan J, Hsu CH. User allocation-aware edge cloud
  placement in mobile edge computing. {\it Software: Practice and Experience}
  2020\string; 50(5)\string: 489-502.
\newblock \href {\doibase 10.1002/spe.2685} {doi: 10.1002/spe.2685}

\bibitem{Badri2020TPDS}
{Badri} H, {Bahreini} T, {Grosu} D, {Yang} K. Energy-Aware Application
  Placement in Mobile Edge Computing: A Stochastic Optimization Approach. {\it
  IEEE Transactions on Parallel and Distributed Systems} 2020\string;
  31(4)\string: 909-922.

\bibitem{gupta2017ifogsim}
Gupta H, Vahid~Dastjerdi A, Ghosh SK, Buyya R. iFogSim: A toolkit for modeling
  and simulation of resource management techniques in the Internet of Things,
  Edge and Fog computing environments. {\it Software: Practice and Experience}
  2017\string; 47(9)\string: 1275--1296.

\bibitem{crawdad}
crawdad. A Community Resource for Archiving Wireless Data At Dartmouth;  2009.
\newblock \url{http://crawdad.org/epfl/mobility/}.

\bibitem{WangTMC}
{Wang} S, {Guo} Y, {Zhang} N, {Yang} P, {Zhou} A, {Shen} XS. Delay-aware
  Microservice Coordination in Mobile Edge Computing: A Reinforcement Learning
  Approach. {\it IEEE Transactions on Mobile Computing} 2019\string: 1-1.

\bibitem{Wang2019ToN}
{Wang} S, {Urgaonkar} R, {Zafer} M, {He} T, {Chan} K, {Leung} KK. Dynamic
  Service Migration in Mobile Edge Computing Based on Markov Decision Process.
  {\it IEEE/ACM Transactions on Networking} 2019\string; 27(3)\string:
  1272-1288.

\bibitem{Wang2015}
{Wang} S, {Urgaonkar} R, {Zafer} M, {He} T, {Chan} K, {Leung} KK. Dynamic
  service migration in mobile edge-clouds. In: 2015 IFIP Networking Conference
  (IFIP Networking). ; 2015\string: 1-9.

\bibitem{Samanta2020}
{Samanta} A, {Tang} J. Dyme: Dynamic Microservice Scheduling in Edge Computing
  Enabled IoT. {\it IEEE Internet of Things Journal} 2020\string; 7(7)\string:
  6164-6174.

\bibitem{Samanta2019}
{Samanta} A, {Li} Y, {Esposito} F. Battle of Microservices: Towards
  Latency-Optimal Heuristic Scheduling for Edge Computing. In: 2019 IEEE
  Conference on Network Softwarization (NetSoft). ; 2019\string: 223-227.

\bibitem{Poularakis2019InfoCom}
{Poularakis} K, {Llorca} J, {Tulino} AM, {Taylor} I, {Tassiulas} L. Joint
  Service Placement and Request Routing in Multi-cell Mobile Edge Computing
  Networks. In: IEEE INFOCOM 2019 - IEEE Conference on Computer Communications.
  ; 2019\string: 10-18.

\bibitem{Pasteris2019InfoCom}
{Pasteris} S, {Wang} S, {Herbster} M, {He} T. Service Placement with Provable
  Guarantees in Heterogeneous Edge Computing Systems. In: IEEE INFOCOM 2019 -
  IEEE Conference on Computer Communications. ; 2019\string: 514-522.

\bibitem{ZHANG2019111}
Zhang C, Zheng Z. Task migration for mobile edge computing using deep
  reinforcement learning. {\it Future Generation Computer Systems} 2019\string;
  96\string: 111 - 118.
\newblock \href {\doibase https://doi.org/10.1016/j.future.2019.01.059} {doi:
  https://doi.org/10.1016/j.future.2019.01.059}

\bibitem{Wan2019}
{Wan} L, {Sun} L, {Kong} X, {Yuan} Y, {Sun} K, {Xia} F. Task-Driven Resource
  Assignment in Mobile Edge Computing Exploiting Evolutionary Computation. {\it
  IEEE Wireless Communications} 2019\string; 26(6)\string: 94-101.

\bibitem{Gao2019}
{Gao} B, {Zhou} Z, {Liu} F, {Xu} F. Winning at the Starting Line: Joint Network
  Selection and Service Placement for Mobile Edge Computing. In: IEEE INFOCOM
  2019 - IEEE Conference on Computer Communications. ; 2019\string: 1459-1467.

\bibitem{Yu2018GlobeCom}
{Yu} N, {Xie} Q, {Wang} Q, {Du} H, {Huang} H, {Jia} X. Collaborative Service
  Placement for Mobile Edge Computing Applications. In: 2018 IEEE Global
  Communications Conference (GLOBECOM). ; 2018\string: 1-6.

\bibitem{Ouyang2019INfoCOm}
{Ouyang} T, {Li} R, {Chen} X, {Zhou} Z, {Tang} X. Adaptive User-managed Service
  Placement for Mobile Edge Computing: An Online Learning Approach. In: IEEE
  INFOCOM 2019 - IEEE Conference on Computer Communications. ; 2019\string:
  1468-1476.

\bibitem{Wu2019ICWS}
{Wu} H, {Deng} S, {Li} W, et al. Mobility-Aware Service Selection in Mobile
  Edge Computing Systems. In: 2019 IEEE International Conference on Web
  Services (ICWS). ; 2019\string: 201-208.

\bibitem{ghosh2019mobi}
Ghosh S, Mukherjee A, Ghosh SK, Buyya R. Mobi-iost: mobility-aware
  cloud-fog-edge-iot collaborative framework for time-critical applications.
  {\it IEEE Transactions on Network Science and Engineering} 2019.

\bibitem{shi2017maga}
Shi Y, Chen S, Xu X. MAGA: A mobility-aware computation offloading decision for
  distributed mobile cloud computing. {\it IEEE Internet of Things Journal}
  2017\string; 5(1)\string: 164--174.

\bibitem{YU2018722}
Yu F, Chen H, Xu J. DMPO: Dynamic mobility-aware partial offloading in mobile
  edge computing. {\it Future Generation Computer Systems} 2018\string;
  89\string: 722-735.
\newblock \href {\doibase https://doi.org/10.1016/j.future.2018.07.032} {doi:
  https://doi.org/10.1016/j.future.2018.07.032}

\bibitem{Ouyang2018}
{Ouyang} T, {Zhou} Z, {Chen} X. Follow Me at the Edge: Mobility-Aware Dynamic
  Service Placement for Mobile Edge Computing. {\it IEEE Journal on Selected
  Areas in Communications} 2018\string; 36(10)\string: 2333-2345.

\bibitem{wyner1974recent}
Wyner A. Recent results in the Shannon theory. {\it IEEE Transactions on
  information Theory} 1974\string; 20(1)\string: 2--10.

\bibitem{antennasearch}
Antennasearch. Antenna Distribution;  2020.
\newblock \url{www.antennasearch.com}.

\bibitem{Ding2019}
{Ding} Z, {Xu} J, {Dobre} OA, {Poor} HV. Joint Power and Time Allocation for
  NOMA–MEC Offloading. {\it IEEE Transactions on Vehicular Technology}
  2019\string; 68(6)\string: 6207-6211.

\bibitem{park2006comon}
Park K, Pai VS. CoMon: a mostly-scalable monitoring system for PlanetLab. {\it
  ACM SIGOPS Operating Systems Review} 2006\string; 40(1)\string: 65--74.

\bibitem{WANG2019JPDC}
Wang S, Zhao Y, Xu J, Yuan J, Hsu CH. Edge server placement in mobile edge
  computing. {\it Journal of Parallel and Distributed Computing} 2019\string;
  127\string: 160 - 168.
\newblock \href {\doibase https://doi.org/10.1016/j.jpdc.2018.06.008} {doi:
  https://doi.org/10.1016/j.jpdc.2018.06.008}

\bibitem{Xu2020TSUSC}
{Xu} M, {N. Toosi} A, {Buyya} R. A Self-adaptive Approach for Managing
  Applications and Harnessing Renewable Energy for Sustainable Cloud Computing.
  {\it IEEE Transactions on Sustainable Computing} 2020\string: 1-1.
\newblock \href {\doibase 10.1109/TSUSC.2020.3014943} {doi:
  10.1109/TSUSC.2020.3014943}

\bibitem{zhao2019optimal}
Zhao L, Wang J, Liu J, Kato N. Optimal edge resource allocation in IoT-based
  smart cities. {\it IEEE Network} 2019\string; 33(2)\string: 30--35.

\bibitem{Anton2012CCPE}
Beloglazov A, Buyya R. Optimal online deterministic algorithms and adaptive
  heuristics for energy and performance efficient dynamic consolidation of
  virtual machines in Cloud data centers. {\it Concurrency and Computation:
  Practice and Experience} 2012\string; 24(13)\string: 1397-1420.
\newblock \href {\doibase https://doi.org/10.1002/cpe.1867} {doi:
  https://doi.org/10.1002/cpe.1867}

\end{thebibliography}

\newpage
\section*{Author Biography}

\begin{biography}{\includegraphics[width=65pt]{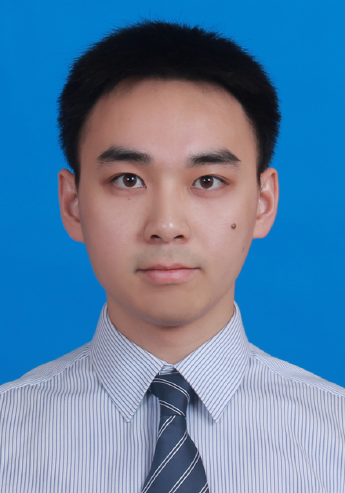}}{\textbf{Minxian Xu} is currently an assistant professor at Shenzhen Institutes of Advanced Technology, Chinese Academy of Sciences. He received the BSc degree	in 2012 and the MSc degree in 2015, both in software engineering from University of Electronic Science and Technology of China. He obtained his PhD degree from the University of Melbourne in 2019. His research interests include resource scheduling and optimization in cloud computing. He has co-authored 20+ peer-reviewed papers published in prominent international journals and conferences, such as 
CSUR, T-SUSC, T-ASE, JPDC, JSS, ICSOC. His Ph.D. Thesis was awarded the 2019 IEEE TCSC Outstanding Ph.D. Dissertation Award. More information can be found at: minxianxu.info.}
\end{biography}
 
 \begin{biography}{\includegraphics[width=65pt]{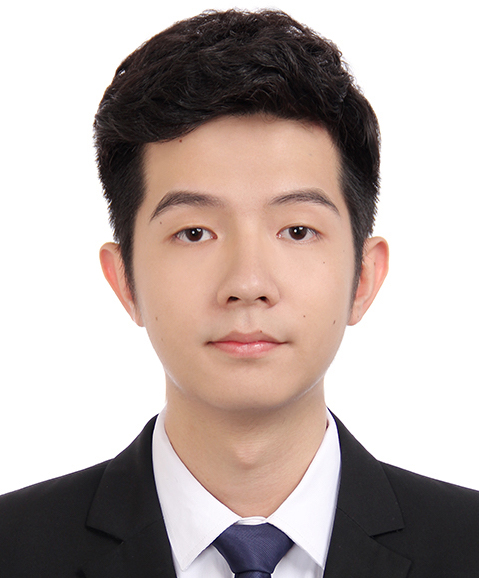}}{\textbf{Qiheng Zhou} received his BSc degree from Sun Yat-sen University. He is currently a master student at National University of Singapore. His research interests include cloud computing and blockchain.}
\end{biography}

\vspace{1.2cm}

\begin{biography}{\includegraphics[width=65pt]{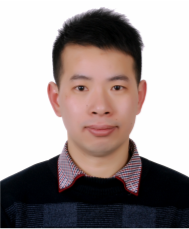}}{\textbf{Huaming Wu} received the B.E. and M.S. degrees from Harbin Institute of Technology, China in 2009 and 2011, respectively, both in electrical engineering. He received the Ph.D. degree with the highest honor in computer science at Freie Universit\"at Berlin, Germany in 2015. He is currently an associate professor in the Center for Applied Mathematics, Tianjin University, China. His research interests include model-based evaluation, wireless and mobile network systems, mobile cloud computing and deep learning.}
\end{biography}
	
	\vspace{0.05cm}
	
\begin{biography}{\includegraphics[width=65pt]{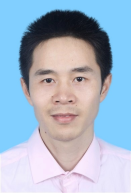}}{\textbf{Weiwei Lin} received his B.S. and M.S. degrees from Nanchang University in 2001 and 2004, respectively, and the PhD degree in Computer Application from South China University of Technology in 2007. He has been serving as visiting scholar at Clemson University from 2016 to 2017. Currently, he is a professor in the School of Computer Science and Engineering, South China University of Technology. His research interests include distributed systems, cloud computing, big data computing and AI application technologies. He has published more than 100 papers in refereed journals and conference proceedings. He has been the reviewers for many international journals, including TPDS, TC, TMC, TCYB, TSC, TCC, etc. }
\end{biography}

\begin{biography}{\includegraphics[width=66pt,height=86pt]{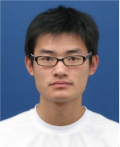}}{\textbf{Kejiang Ye} received his BSc and PhD degree in Computer Science from Zhejiang University in 2008 and 2013, respectively. He was also a joint PhD student at The University of Sydney from 2012 to 2013. After graduation, he works as Post-Doc Researcher at Carnegie Mellon University from 2014 to 2015 and Wayne State University from 2015 to 2016. He is currently a Professor at Shenzhen Institutes of Advanced Technology, Chinese Academy of Science. His research interests focus on the performance, energy, and reliability of cloud computing and network systems.}
\end{biography}

\begin{biography}{\includegraphics[width=66pt,height=86pt]{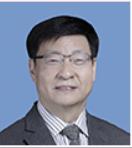}}{\textbf{Chengzhong Xu} (Fellow, IEEE) is the Dean of Faculty of Science and Technology and the Interim Director of Institute of Collaborative Innovation, University of Macau, and a Chair Professor of Computer and Information Science. Dr. Xu’s main research interests lie in parallel and distributed computing and cloud computing, in particular, with an emphasis on resource management for system’s performance, reliability, availability, power efficiency, and security, and in big data and data-driven intelligence applications in smart city and self-driving vehicles. He published two research monographs and more than 300 peer-reviewed papers in journals and conference proceedings; his papers received about 10K citations with an H-index of 52. He serves or served on a number of journal editorial boards, including IEEE Transactions on Computers (TC), IEEE Transactions on Cloud Computing (TCC), IEEE Transactions on Parallel and Distributed Systems (TPDS), Journal of Parallel and Distributed Computing (JPDC), Science China: Information Science and ZTE Communication. Dr. Xu has been the Chair of IEEE Technical Committee on Distributed Processing (TCDP) since 2015. He obtained BSc and MSc degrees from Nanjing University in 1986 and 1989 respectively, and a PhD degree from the University of Hong Kong in 1993, all in Computer Science and Engineering. }
\end{biography}

\end{document}